\documentclass[journal]{IEEEtran}
\usepackage{amsmath,amsfonts,amsthm} 
\usepackage{bbm}

\newtheorem{theorem}{Theorem}
\newtheorem{definition}{Definition}
\newtheorem{lemma}{Lemma}
\newtheorem{assumption}{Assumption}
\newtheorem{proposition}{Proposition}
\newtheorem{corollary}{Corollary}
\newtheorem{remark}{Remark}

\usepackage[noadjust]{cite}

\usepackage{accents}

\newcommand{\R}{\mathbb{R}}
\newcommand{\N}{\mathbb{N}}

\newcommand{\BlockNorm}[1]{\left[\kern-0.15em\left[ \left. {#1} 
		\right]\kern-0.15em\right]
	\right.\;}

\usepackage{graphicx} 
\usepackage{subfigure}

\usepackage{color}
\usepackage{tikz-timing}
\usepackage{tikz}
\usetikzlibrary{arrows,positioning,automata}
\usetikzlibrary{arrows,through}

\begin{document}
%
\title{Lyapunov Criterion for Stochastic Systems and Its Applications in Distributed Computation}
%
%
%
\author{Yuzhen~Qin,~\IEEEmembership{Student~Member,~IEEE,}
	 Ming~Cao,~\IEEEmembership{Senior~Member~IEEE,}\\
	 and~Brian D. O. Anderson,~\IEEEmembership{Life~Fellow,~IEEE} 
\thanks{Y. Qin and M. Cao are with the Institute of Engineering and Technology, Faculty of Science and Engineering, University of Groningen, Groningen, the Netherlands (\{y.z.qin, m.cao\}@rug.nl). B.D.O. Anderson is with School of Automation, Hangzhou Dianzi University, Hangzhou, 310018, China, and Data61-CSIRO and Research School of Engineering, Australian	National University, Canberra, ACT 2601, Australia (Brian.Anderson@anu.edu.au).

The work of Cao was supported in part by the European Research Council (ERC-CoG-771687) and the Netherlands Organization for Scientific Research (NWO-vidi-14134). The work of B.D.O. Anderson was supported by the Australian Research Council (ARC) under grants DP-130103610 and DP-160104500, and by Data61-CSIRO.
}
}


\maketitle

\begin{abstract}
This paper presents new sufficient conditions for {convergence and} asymptotic or exponential stability of a stochastic discrete-time system, under which the constructed Lyapunov function always decreases in expectation along the system's solutions after a finite number of steps, but without necessarily strict decrease at every step, in contrast to the classical stochastic Lyapunov theory. As the first application of this new Lyapunov criterion,  we look at the product of any random sequence of stochastic matrices, including those with zero diagonal entries, and obtain sufficient conditions to ensure the product almost surely converges to a matrix with identical rows; we also show that the rate of convergence can be exponential under additional conditions. As the second application, we study a distributed network algorithm for solving linear algebraic equations. We relax existing conditions on the network structures, while still guaranteeing the equations are solved asymptotically.
\end{abstract}

%
\IEEEpeerreviewmaketitle

\section{Introduction}
Stability analysis for stochastic dynamical systems has always been an active research field. Early works have shown that stochastic Lyapunov functions play an important role, and to use them for discrete-time systems, a standard procedure is to show that they decrease in \emph{expectation} at every time step \cite{Kushner_1972,Kushner_1971,Kushner_1965,beutler1973two}.  Properties of supermartingales and LaSalle's arguments are critical to establish the related proofs. However, most of the stochastic stability results are built upon a crucial assumption, which requires that the stochastic dynamical system under study is Markovian (see e.g., \cite{Kushner_1972,Kushner_1971,Kushner_1965,Khasminskii}), and very few of them have reported bounds for the convergence speed.

More recently, with the fast development of network algorithms, more and more distributed computational processes are carried out in networks of coupled computational units. Such dynamical processes are usually modeled by stochastic discrete-time dynamical systems since they are usually under inevitable influences from random changes of network structures {\cite{porfiri2007consensus,Tahbaz-Salehi_2010,lee2017stochastic,nedic2016stochastic}}, communication delay and noise {\cite{nilsson1998stochastic,yang2011filtering,wu2011consensus}}, and asynchronous updating events {\cite{tsitsiklis1986distributed,lee2016asynchronous}}. So there is great need in further developing Lyapunov theory for stochastic dynamical systems, in particular in the setting of network algorithms for distributed computation. And this is exactly the aim of this paper.

We aim at further developing the Lyapunov criterion for stochastic discrete-time systems. Motivated by the concept of  \textit{finite-step Lyapunov functions} for deterministic systems \cite{Aeyels,Geiselhart,Gielen}, we propose to define a \textit{finite-step stochastic Lyapunov function}, which decreases in expectation, not necessarily at every step, but after a finite number of steps. The associated new Lyapunov criterion not only enlarges the range of choices of candidate Lyapunov functions but also implies that the systems that it can be used to analyze do not need to be Markovian. An additional advantage of using this new criterion is that we are enabled to construct conditions to guarantee  exponential convergence and estimate convergence rates.

We then apply the finite-step stochastic Lyapunov function to study two distributed computation problems arising in some popular network algorithmic settings. In distributed optimization \cite{optimization,optimization_1} and other distributed coordination algorithms \cite{Reaching_consensus,Agreeing_asynchronously,Tahbaz-Salehi_2008,Tahbaz-Salehi_2010}, one frequently encounters the need to prove convergence of inhomogeneous Markov chains, or equivalently the convergence of backward products of random sequences of stochastic matrices $\{W(k)\}$.  Most of the existing results assume exclusively that all the $W(k)$ in the sequence have all positive diagonal entries, see e.g., \cite{Wu_Chai,Hendrickx,Liu_ji}. This assumption simplifies the analysis of convergence significantly; moreover, without this assumption, the existing results do not always hold. For example, from \cite{Tahbaz-Salehi_2008,Tahbaz-Salehi_2010} one knows that the product of $W(k)$ converges to a rank-one matrix almost surely if exactly one of the eigenvalues of the expectation of $W(k)$ has the modulus of one, which can be violated if $W(k)$ has zero diagonal elements. Note also that most of  the existing results  are confined to {special} random sequences, e.g., independently distributed sequences \cite{Tahbaz-Salehi_2008}, stationary ergodic sequences \cite{Tahbaz-Salehi_2010}, or independent sequences \cite{Touri,Touri_1}. Using the new Lyapunov criterion in this paper, we work on more general classes of random sequences of stochastic matrices without the assumption of non-zero diagonal entries. We show that if there exists a fixed length such that the product of any successive subsequence of matrices of this length has the \textit{scrambling} property (a standard concept, but it will be defined subsequently) with positive probability, the convergence to a rank-one matrix for the infinite product can be guaranteed almost surely. We also prove that the convergence can be exponentially fast if this probability is lower bounded by some positive number, and the greater the lower bound is, the faster the convergence becomes. For some particular random sequences, we further relax this ``scrambling" condition. If the random sequence is driven by a \textit{stationary} process, the almost sure convergence can be ensured as long as the product of any successive subsequence of finite length has positive probability to be indecomposable and aperiodic (SIA). The exponential convergence rate follows without other assumptions if the random process that governs the evolution of the sequence is a \textit{stationary ergodic} process.

As the second application of the finite-step stochastic Lyapunov functions, we investigate a distributed algorithm for solving linear algebraic equations of the form $Ax=b$. The equations are solved in parallel by $n$ agents, each of whom just knows a subset of the rows of the matrix $[A,b]$. Each agent recursively updates its estimate of the solution using the current estimates from its neighbors. Recently several solutions under different sufficient conditions have been proposed \cite{Liu_Ji_2017,Mou_2016,equation_Mou}, and in particular in \cite{equation_Mou}, the sequence of the neighbor relationship graphs $\mathcal{G}(k)$ is required to be repeated jointly strongly connected. We show that a much weaker condition is sufficient to solve the problem almost surely, namely the algorithm in \cite{equation_Mou} works if there exists a fixed length such that any subsequence of $\{\mathcal{G}(k)\}$ at this length is jointly strongly connected with positive probability.

The remainder of this paper is organized as follows. In Section II, we define the finite-step stochastic Lyapunov functions. Products of random sequences of stochastic matrices are studied in Section III; in Section IV we look into in particular the asynchronous implementation issues as an application of Section III. Finally, we study in Section V a distributed approach for solving linear equations. Brief concluding remarks appear in Section VI. 

{\textit{Notation}: Throughout this paper, $\N_0$ denotes the sets of non-negative integers, $\N$ the collection of positive integers, and $\R^q$ the real $q$-dimensional vector space. Moreover, we let $\mathbf{1}$ be the vector consisting of all ones, and let $\mathbf N=\{1,2,\dots,n\}$. Given a vector $x\in \R^n$, $x^i$ denotes the $i$th element of $x$. Let $\left\|\cdot\right\|$, $p\ge 1$, be any $p$-norm. A continuous function $h(x): [0,a)\to[0,\infty)$ is said to belong to class $\mathcal K$ if it is strictly increasing and $h(0)=0$. For any two events $A,B$, the conditional probability $\Pr[A|B]$ denotes the probability of $A$ given $B$. }

\section{Finite-Step Stochastic Lyapunov Functions}

 Consider a stochastic discrete-time system described by 	
	\begin{align}\label{model:main}
	&x_{k+1}=f(x_k,y_{k+1}),   &k\in \N_0,
	\end{align}
where $x_k\in \R^n$, {and $\{y_k: k \in \N\}$ is a $\R^d$-valued stochastic process on a probability space $(\Omega,\mathcal{F},\Pr )$}. Here $\Omega=\{\omega\}$ is the sample space; $\mathcal{F}$ is a set of events which is a $\sigma$-field; $\Pr : \mathcal{F} \to [0,1]$ is a function that assigns probabilities to events; {$y_k$ is a measurable function mapping $\Omega$ into the state space $\Omega_0 \subseteq \R^d$, and for any $\omega\in\Omega$, $\{y_k(\omega):k\in \N\}$ is a realization of the stochastic process $\{y_k\}$ at $\omega$}. Let $\mathcal F_k=\sigma(y_1,\dots,y_k)$ for $k \ge 1$, $\mathcal F_0=\{ \emptyset,\Omega\}$, so that evidently {$\{\mathcal F_k\}, k=1,2,\dots,$ is an increasing sequence of $\sigma$-fields}. Following \cite{reif1999stochastic}, we consider a constant initial condition $x_0\in \R^n$ with probability one. It then can be observed that the solution to \eqref{model:main}, $\{x_k\}$, is a {$\R^n$-valued} stochastic process adapted to $\mathcal F_k$.
The randomness of $y_k$ can be due to various reasons, e.g., stochastic disturbances or noise. Note that \eqref{model:main} becomes a stochastic switching system if $f(x,y )=g_{y}(x)$, where $y$ {maps $\Omega$ into} the set $\Omega_0:=\{1,\dots,p\}$, and $\{g_p(x):\mathbb R^n \to \mathbb R^n,p\in  \Omega_0\}$ is a given family of functions.
 
A point $x^*$ is said to be an \textit{equilibrium} of system \eqref{model:main} if $f(x^*,y)=x^*$ for any $y\in \Omega_0$. Without loss of generality, we assume that the origin $x=0$ is an equilibrium.
 Researchers have been interested in studying the limiting behavior of the solution $\{x_k\}$, i.e., when and to where $x_k$ converges as $k \to \infty$. Most noticeably, Kushner developed classic results on stochastic stability by employing stochastic Lyapunov functions \cite{Kushner_1972,Kushner_1971,Kushner_1965}. We introduce some related definitions before recalling some Kushner's results. Following \cite[Sec. 1.5.6]{Gallager} and \cite{Bitmead}, we first define convergence and exponential convergence of a sequence of random variables.

\begin{definition}[Convergence]\label{def:convergence}
	 A random sequence $\{x_k\in \R^n\}$  in a sample space $\Omega$ converges to a random variable $x$ \emph{almost surely} if $
		\Pr \left[ \omega \in \Omega:\lim_{k\to \infty}\|x_k(\omega)-x \|=0 \right]=1.
		$
The convergence is said to be exponentially fast with a rate {no slower than} $\gamma^{-1}$ for some $\gamma>1$ independent of $\omega$ if {$\gamma^k \|x_k-x\|$  almost surely converges to $y$ for some finite $y\ge 0$.} Furthermore, let $\mathcal D \subset\R^n$ be a set; a random sequence $\{x_k\}$ is said to converge to $\mathcal D$ almost surely if $\Pr\left[ \omega \in \Omega:\lim_{k\to \infty}{\rm dist} (x_k(\omega), \mathcal{D} )=0\right]=1,$ where ${\rm dist\;}(x,\mathcal D):=\inf_{y\in \mathcal D}\|x-y\|$.
\end{definition} 

Here ``almost surely" is exchangeable with ``with probability one", and we sometimes use the shorthand notation ``a.s.".  {We now introduce some stability concepts for stochastic discrete-time systems analogous to those in \cite{Khasminskii} and \cite{nishimura2016conditions} for continuous-time systems{\footnote[1]{Note that 1) and 2) of Definition \ref{de_stability_inpro} follow from the definitions in \cite[Chap. 5]{Khasminskii}, in which an arbitrary initial time $s$ rather than just $0$ is actually considered.  We define 3) following the same lines as 1) and 2). In Definition \ref{de_stability_as}, 1) follows from the definitions in \cite{nishimura2016conditions}, and we define 2) following the same lines as 1). }}.

\begin{definition}\label{de_stability_inpro}
	The origin of (1) is said to be: 
	
	 \emph{1) stable in probability} if $\lim\nolimits_{x_0\to 0} \Pr \left[ \sup\nolimits_{k\in \N}\|x_k\|>\varepsilon \right]=0$ for any $\varepsilon>0$; 
	
	\emph{2) asymptotically stable in probability} if it is stable in probability and moreover
	$\lim\nolimits_{x_0\to 0} \Pr \left[ \lim\nolimits_{k\to \infty}\|x_k\|=0 \right]=1$;
	
	\emph{3) exponentially stable in probability} if for some $\gamma >1$ {independent of $\omega$}, $\lim\nolimits_{x_0\to 0} \Pr \left[ \lim\nolimits_{k\to \infty}\|\gamma^k x_k\|=0 \right]=1$;
\end{definition}

\begin{definition}\label{de_stability_as}
	For a set $\mathcal Q \subseteq \mathbb R^n$ containing the origin, the origin of (1) is said to be:
	
	 \emph{1) locally a.s. asymptotically stable in $\mathcal Q$ (globally a.s. asymptotically stable, respectively)} if starting from $x_0\in \mathcal Q$ ($x_0\in \mathbb R^n$, respectively) all the sample paths $x_k$ stay in $\mathcal Q$ ($\mathbb R^n$, respectively) for all $k\ge 0$ and converge to the origin almost surely; 
	 
	 \emph{2) locally a.s. exponentially stable  in $\mathcal Q$ (globally a.s.  exponentially stable, respectively)}  if it is locally (globally, respectively) a.s. asymptotically stable and the convergence is exponentially fast.
\end{definition}
}

  Now let us recall some Kushner's results on {convergence and stability}, where stochastic Lyapunov functions have been used.
\begin{lemma}[Asymptotic Convergence and Stability] \label{lemma_asym_stab}
	 For the  stochastic discrete-time system \eqref{model:main}, let $\{x_k\}$ be a Markov process. Let $V:\mathbb{R}^n\to \mathbb{R}$ be a {continuous positive definite and radially unbounded function}. Define the set $\mathcal Q_{\lambda}:=\{x:0 \le V(x)<\lambda\}$ for some $\lambda>0$, and assume that
	 \begin{equation}\label{asym_condi}
	 \mathbb{E}\left[ V\left(x_{k+1}\right)|x_k\right]-V\left(x_k\right)\le -\varphi(x_k), \forall k,
	 \end{equation}
	 where {$\varphi: \mathbb R^n\to \mathbb R$ is continuous and satisfies $\varphi(x)\ge 0$ for any $x \in\mathcal Q_\lambda$}. {Then the following statements apply: 
	 	
	 	\emph{i)} for any initial condition $x_0 \in \mathcal Q_\lambda$, $x_k$ converges to $\mathcal{D}_1:=\{x \in \mathcal Q_{\lambda}:\varphi(x)=0\}$ with probability  at least $1-V(x_0)/\lambda$ \cite{Kushner_1965};
	 	 
	 	\emph{ii)} if moreover $\varphi(x)$ is positive definite on $\mathcal Q_\lambda$, and $h_1\left(\| s \|\right) \le V(s)\le h_2\left(\| s \|\right)$ for two class $\mathcal{K}$ functions $h_1$ and $h_2$, then $x=0$ is asymptotically stable in probability \cite{Kushner_1965}, \cite[Theorem 7.3]{kushner2014partial}.
}
\end{lemma}

\begin{lemma}[Exponential Convergence and Stability] \label{lemma_expone_stab}
	For the  stochastic discrete-time system \eqref{model:main}, let $\{x_k\}$ be a Markov process. Let $V:\mathbb{R}^n\to \mathbb{R}$ be a continuous nonnegative function.
	Assume that
	\begin{align}\label{expone_condi}
	&\mathbb{E}\left[ V\left(x_{k+1}\right)|x_k\right]-V\left(x_k\right)\le -\alpha V(x_k), &0<\alpha <1.
	\end{align}
	Then the following statements apply: 
	
	\emph{i)} {for any given $x_0$}, $V(x_k)$ almost surely converges to $0$ exponentially fast with a rate no slower than $1-\alpha$ \cite[Th. 2, Chap. 8]{Kushner_1971}, \cite{kushner2014partial};
	
	\emph{ii)} if moreover $V$ satisfies $c_1 \|x\|^a \le V(x)\le c_2 \|x\|^a$ for some $c_1,c_2, a >0$, then $x=0$ is globally a.s. exponentially stable \cite[Theorem 7.4]{kushner2014partial}. 
\end{lemma}

To use these two lemmas to prove asymptotic (or exponential) stability for a stochastic system, the critical step is to find a stochastic Lyapunov function such that \eqref{asym_condi} (respectively, \eqref{expone_condi}) holds. However, it is not always obvious how to construct such a stochastic Lyapunov function. We use the following toy example to illustrate this point.

\textit{Example 1.} Consider a randomly switching system described by $x_k = A_{y_k} x_{k-1}$, where $y_k$ is the switching signal taking values in a finite set $\mathcal{P}:=\{1,2,3\}$,and
\begin{align*}
A_1=\left[ {\begin{array}{*{20}{c}}
	{0.2}&{0}\\
	{0}&{1}
	\end{array}} \right], A_2=\left[ {\begin{array}{*{20}{c}}
	{1}&{0}\\
	{0}&{0.8}
	\end{array}} \right], A_3=\left[ {\begin{array}{*{20}{c}}
	{1}&{0}\\
	{0}&{0.6}
	\end{array}} \right].
\end{align*}
  The stochastic process $\{y_k\}$ is described by a Markov chain with initial distribution $v=\{v_1,v_2,v_3\}$. The transition probabilities are described by a transition matrix 
  \begin{align*}
  \pi=\left[ {\begin{array}{*{20}{c}}
  	{0}&{0.4}&{0.6}\\
  	{1}&{0}&{0}\\
  	{1}&{0}&{0}
  	\end{array}} \right],
  \end{align*}
  whose $ij$th element is defined by $\pi_{ij}=\Pr[y_{k+1}=j|y_k=i]$. {Since $\{y_k\}$ is not  independent and identically distributed, the process $\{x_k\}$ is not Markovian. Nevertheless, we might }conjecture that the origin is globally a.s. exponentially stable. In order to try to prove this, we might choose a stochastic Lyapunov function candidate $V(x)=\left\| x \right\|_\infty$, but the existing results introduced in Lemma \ref{lemma_expone_stab} cannot be used since $\{x_k\}$ is not Markovian. Moreover, by calculation we observe that  $\mathbb{E}\left[ {\left. {V\left( {x_{k+1}} \right)} \right| x_k,y_k} \right] \le V\left( {x_k} \right)$ for any $y_k$, which implies that  \eqref{expone_condi} {is not necessarily satisfied}.  Thus $V(x)$ is not an appropriate stochastic Lyapunov function for which Lemma \ref{lemma_expone_stab} can be applied. As it turns out however, the same $V(x)$ can be used as a Lyapunov function to establish exponentially stability via the alternative criterion set out subsequently. \hfill   $ \Box $ 

It is difficult, if not impossible, to construct a stochastic Lyapunov function, especially when the state of the system is not Markovian. So it is of great interest to generalize the results in Lemmas \ref{lemma_asym_stab} and \ref{lemma_expone_stab} such that the range of choices of candidate Lyapunov functions can be enlarged. For deterministic systems, Aeyels et al. have introduced a new Lyapunov criterion to study asymptotic stability of continuous-time systems \cite{Aeyels}; {a} similar criterion has also been obtained for discrete-time systems, and the Lyapunov functions satisfying this criterion are called \emph{finite-step Lyapunov functions} \cite{Geiselhart,Gielen}. A common feature of these works is that the Lyapunov function is required to decrease along the system's solutions after a finite number of steps, but not necessarily at every step. We now use this idea to construct stochastic finite-step Lyapunov functions, a task which is much more challenging compared to the deterministic case due to the uncertainty present in stochastic systems. The tools for analysis are totally different from what are used for deterministic systems.  We will exploit supermartingales and their convergence property, as well as the Borel-Cantelli Lemma{; these concepts} are introduced in the two following lemmas.

 \begin{lemma}[{\cite[Sec. 5.2.9]{prob_book}}]\label{supermartingale}
	Let the sequence $\{X_k\}$ be a nonnegative \emph{supermartingale} with respect to $\mathcal{F}_k=\sigma(X_1,\dots,X_k)$, {i.e., suppose:} (i) $\mathbb{E}X_n < \infty$; (ii) $X_k \in \mathcal{F}_k$ for all $k$;  (iii) $\mathbb{E}\left( {\left. {{X_{k + 1}}} \right|{\mathcal{F}_k}} \right)  \le {X_k}$. {Then} there exists some random $X$ such that $X_k\stackrel{a.s.}{\longrightarrow} X, k \to \infty$,
	and $\mathbb{E}X \le \mathbb{E}X_0$. 
\end{lemma}

{\begin{lemma}[{Borel-Cantelli Lemma, \cite[P.192]{Kushner_1971}}] \label{Borel-Cantelli}
	{Let $\{X_k\}$ be a  nonnegative random sequence. If $\sum_{k=0}^{\infty} \mathbb{E} X_k<\infty$, then $X_k \stackrel{a.s.}{\longrightarrow}0$.}
\end{lemma}}

We are now ready to present our first main result {on stochastic convergence and stability}.

\begin{theorem}\label{stability_finite}
	For the stochastic discrete-time system \eqref{model:main},	let
	$V:\mathbb{R}^n\to \mathbb{R}$ be a {continuous nonnegative and radially unbounded function}. Define the set $\mathcal Q_{\lambda}:=\{x: V(x)<\lambda\}$ for some $\lambda>0$, and assume that
	
	\emph{a)} $\mathbb{E}\left[ V\left(x_{k+1}\right)|\mathcal F_k\right]-V\left(x_k\right)\le 0$ for any $k$ such that $x_k \in \mathcal {Q}_\lambda$;
		
	\emph{b)}  there is an integer $T \ge 1$, independent of $\omega$, such that for any $k $, $\mathbb{E}\left[ V\left(x_{k+T}\right)|\mathcal{F}_k\right]-V\left(x_k\right)\le -\varphi(x_k)$, {where $\varphi: \mathbb R^n\to \mathbb R$ is continuous and satisfies $\varphi(x)\ge 0$ for any $x \in\mathcal Q_\lambda$}.	\\		
	 Then the following statements apply: 
	 
	 \emph{i)} for any initial condition $x_0 \in \mathcal Q_\lambda$, $x_k$ converges to $\mathcal{D}_1:=\{x \in \mathcal Q_{\lambda}:\varphi(x)=0\}$ with probability  at least $1-V(x_0)/\lambda$; 
	 
	 \emph{ii)} if moreover $\varphi(x)$ is positive definite on $\mathcal Q_\lambda$, and $h_1\left(\| s \|\right) \le V(s)\le h_2\left(\| s \|\right)$ for two class $\mathcal{K}$ functions $h_1$ and $h_2$, then $x=0$ is asymptotically stable in probability.
\end{theorem}
\begin{proof}
	{Before proving i) and ii), we first show that starting from $x_0\in \mathcal Q_\lambda$ the sample paths $x_k(\omega)$ stay in $\mathcal Q_\lambda$ with probability at least $1-V(x_0)/\lambda$ if Assumption a) is satisfied. This has been proven in \cite[p. 196]{Kushner_1971} by showing that \begin{align}\label{Pr_leave}
		\Pr\left[ { \mathop {\sup }\nolimits_{k\in \N} {V\left( {{x_k}} \right) \ge \lambda }} \right] \le {{V\left( {{x_0}} \right)}}/{\lambda }.
			\end{align}}			
	 Let $\bar \Omega$ be a subset of the sample space $\Omega$ such that for any $\omega \in \bar \Omega$, $x_k(\omega)\in \mathcal Q_\lambda$ for all $k$.	
	Let $J$ be the smallest $k\in \N$ (if it exists) such that $V(x_k) \ge \lambda$. Note that, this integer $J$ does not exist when $x_k(\omega)$ stays in $\mathcal Q_\lambda$ for all $k$, i.e., when $\omega \in \bar{\Omega}$.
	
	{We first prove i) by showing that the sample paths staying the $\mathcal Q_\lambda$ converge to $\mathcal{D}_1$ with probability one, i.e., $\Pr[x_k\to \mathcal{D}_1| \bar \Omega]=1$.} Towards this end, define a new function $\tilde{\varphi}(x)$ such that $\tilde{\varphi}(x)=\varphi(x)$ for $x \in \mathcal{Q}_{\lambda}$, and $\tilde{\varphi}(x)=0$ for $x \notin \mathcal{Q}_{\lambda}$. Define another random process $\{\tilde z_k\}$. {If $J$ exists}, when $J>T$ let 
	\begin{align*}
	&\tilde z_k=x_k,k<J-T,
	&\tilde z_k=\epsilon, k\ge J-T,
	\end{align*}
	where $\epsilon$ satisfies $V(\epsilon)=\tilde{\lambda}>\lambda$; 
	when $J \le T$, let $\tilde z_k=\epsilon$ for any $k \in \N_0$. {If $J$ does not exist}, we let $\tilde z_k=x_k$ for all $k\in\N_0$.
	Then it is immediately clear that $\mathbb{E}\left[ V \left( \tilde z_{k+T} \right) | \mathcal{F}_k \right] - V \left(\tilde z_k\right) \le -\tilde \varphi(\tilde z_k) \le 0$. By taking the expectation on both sides of this inequality, we obtain
	\begin{align}\label{expectation:T}
	\mathbb{E}\left[ V\big (\tilde z_{k+T} \big) \right] -\mathbb{E} V \big(\tilde z_k\big)\le -\mathbb{E} \tilde \varphi\big(\tilde z_k\big), k\in \N_0.
	\end{align}
	{For any $k\in\N$, there is a pair $p,q\in\N_0$ such that $k=pT+q$. It follows from \eqref{expectation:T} that
		\begin{align*}
		\mathbb{E} \left[ V\big (\tilde z_{pT+j} \big) \right] -\mathbb{E} V \big(\tilde z_{(p-1)T+j}\big)\le -\mathbb{E} \tilde \varphi\big(\tilde z_{(p-1)T+q}\big),\;\;&\\
		j=1,\dots,q;&\\
		\mathbb{E} \left[ V\big (\tilde z_{iT+m} \big) \right] -\mathbb{E} V \big(\tilde z_{(i-1)T+m}\big)\le -\mathbb{E} \tilde \varphi\big(\tilde z_{(i-1)T+m}\big),&\\
		i=1,\dots,p-1,m=0,\dots,T-1	&
		\end{align*}
		By summing up all the left and right sides of these inequalities respectively for all the $i,j$ and $m$, we have 
		\begin{align}
		\sum_{m=0}^{T-1}\Big( \mathbb{E}\big[ V(\tilde z_{(p-1)T+m} - \mathbb{E} V(\tilde z_{m}\big) \big]  \Big) +\sum_{j=1}^{q}\Big( \mathbb{E}\big[ V(&\tilde z_{pT+j} -\nonumber\\
		\mathbb{E} V(\tilde z_{(p-1)T+j}\big) \big]  \Big) \le -\sum_{i=1}^{k-T}\mathbb{E} \tilde \varphi\big(\tilde z_{i}\big).& \label{expectation:nT}
		\end{align}	
		As $V(x)$ is nonnegative for all $x$,  from \eqref{expectation:T} it is easy to observe that the left side of \eqref{expectation:nT} is greater than $-\infty$ even when $k\to \infty$ since $T$ and $q$ are finite numbers, which implies that $\sum_{i=0}^{\infty} \mathbb{E}  \tilde  \varphi\big(\tilde z_{k}\big)< \infty$. By Lemma \ref{Borel-Cantelli}, ones knows that $\tilde  \varphi\big(\tilde z_{k}\big)\stackrel{a.s.}{\longrightarrow}  0$ as $k\to \infty$.} For $\omega \in \bar \Omega$, one can observe that $\tilde \varphi(x_k(\omega)) = \varphi(x_k(\omega))$ and $\tilde z_k \left(\omega\right)=x_k(\omega)$ according to the definitions of $\tilde \varphi$ and $\{\tilde z_k\}$, respectively.	
	Therefore, $\tilde{\varphi}(\tilde z_k(\omega))=\varphi(x_k(\omega))$ for all $\omega \in \bar \Omega$, and subsequently
	\begin{align*}
	\Pr[\varphi \left(x_k\right)\to 0|\bar \Omega]= \Pr[\tilde{\varphi}
	\left(\tilde z_k\right)\to 0|\bar \Omega]=1.
	\end{align*}
	From the continuity of $\varphi(x)$ it can be seen that $\Pr[x_k\to \mathcal{D}_1| \bar \Omega]=1$. The proof of i) is complete since \eqref{Pr_leave} means that the sample paths stay in $\mathcal Q_\lambda$ with probability at least $1-V(x_0)/\lambda$. 
	
	{Next, we prove ii) in two steps. 
	We first prove that the origin $x=0$ is stable in probability. The inequalities $h_1\left(\| s \|\right) \le V(s)\le h_2\left(\| s \|\right)$ imply that $V(x)=0$ if and only if $x=0$. Moreover,  it follows from $h_1\left(\| s \|\right) \le V(s)$ and the inequality \eqref{Pr_leave} that for any initial condition $x_0 \in \mathcal Q_\lambda$, 
\begin{equation*}
\Pr\left[ { \mathop {\sup }\limits_{k\in \N} {h_1\left( {\|x_k\|} \right) \ge \lambda_1 }} \right] \le \Pr\left[ { \mathop {\sup }\limits_{k\in \N} {V\left( {{x_k}} \right) \ge \lambda_1 }} \right] \le \frac{{V\left( {{x_0}} \right)}}{\lambda_1 }
\end{equation*}
for any $\lambda_1>0$. Since $h_1$ is a class $\mathcal K$ function and thus invertible, it can be observed that $\Pr\left[ { \mathop {\sup }_{k\in \N} { {\|x_k\|}  \ge  h_1^{-1}(\lambda) }} \right]\le V(x_0)/\lambda \le h_2(\|x_0\|)/\lambda$. Then for any $\varepsilon >0$, there holds that  $\lim_{x_0\to 0}\Pr\left[ { \mathop {\sup }_{k\in \N} { {\|x_k\|}  >\varepsilon }} \right]\le \Pr\left[ { \mathop {\sup }_{k\in \N} { {\|x_k\|}  \ge \varepsilon }} \right]=0$, which means that the origin is stable in probability.

Second, we show the probability that $x_k\to 0$ tends to $1$ as $x_0\to 0$. One knows that $\mathcal D_1=\{0\}$ since $\varphi$ is positive definite in $\mathcal Q_\lambda$.  From i) one knows that $x_k$ converges to $x=0$ with probability  at least $1-V(x_0)/\lambda$. Since $V(x)\to 0$ as $x_0\to 0$, there holds that $\lim\nolimits_{x_0\to 0}\Pr \left[ \lim\nolimits_{k\to \infty}\|x_k\|=0 \right]\to 1$. The proof is complete. }
\end{proof}

Particularly, if $\mathcal Q_\lambda$ is positively invariant, i.e., starting from $x_0\in \mathcal Q_\lambda$ all sample paths $x_k$ will stay in $\mathcal Q_\lambda$ for all $k \ge 0$, this corollary follows from  Theorem \ref{stability_finite} straightforwardly.

{
\begin{corollary}\label{Coroll_invariant}
	If $\mathcal Q_\lambda$ is positively invariant w.r.t the system \eqref{model:main} and the assumptions a) and b) in Theorem \ref{stability_finite} are satisfied, then the following statements apply:
	
	\emph{i)} for any initial condition $x_0 \in \mathcal Q_\lambda$, $x_k$ converges to $\mathcal{D}_1$ with probability one; 
	
	\emph{ii)} if moreover $\varphi(x)$ is positive definite on $\mathcal Q_\lambda$, and $h_1\left(\| s \|\right) \le V(s)\le h_2\left(\| s \|\right)$ for two class $\mathcal{K}$ functions $h_1$ and $h_2$, then $x=0$ is locally a.s. asymptotically stable in $\mathcal Q_\lambda$. Furthermore, if $\mathcal Q_\lambda=\mathbb R^n$, then $x=0$ is globally a.s. asymptotically stable.
\end{corollary}}

{The next theorem provides a new criterion for exponential convergence and stability of stochastic systems, relaxing the conditions required by Lemma \ref{lemma_expone_stab}. 

\begin{theorem}\label{stability_expone_finite}
	Suppose {the assumptions a) and b)} of Theorem \ref{stability_finite} are satisfied with the inequality of b) strengthened to 
	\begin{align}\label{finite_decre}
	&\mathbb{E}\left[ V\left(x_{k+T}\right)|\mathcal{F}_k\right]-V\left(x_k\right)\le -\alpha V(x_k), &0<\alpha <1.
	\end{align}
	Then the following statements apply: 
	
	\emph{i)} for any given $x_0 \in \mathcal Q_\lambda$, $V(x_k)$ converges to $0$ exponentially at a rate {no slower than} $(1-\alpha)^{{1}/{T}}$, and $x_k$  converges to $\mathcal D_2:=\{x\in \mathcal Q_\lambda:V(x)=0\}$, with probability at least $1-V(x_0)/\lambda$;
	
	\emph{ii)} if moreover $V$ satisfies that $c_1 \|x\|^a \le V(x)\le c_2 \|x\|^a$ for some $c_1,c_2, a >0$, then $x=0$ is exponentially stable in probability.
\end{theorem}}

\begin{proof}
	{We first prove i). From the proof of Theorem \ref{stability_finite}, we know that the sample paths $x_k$ stay in $\mathcal Q_\lambda$ with probability at least $1-V(x_0)/\lambda$ for any initial condition $x_0\in\mathcal Q_\lambda$ if the assumption a) is satisfied. We next show that for any sample path that always stays in $\mathcal Q_\lambda$, $V(x_k)$ converges to $0$ exponentially fast. Towards this end, we define a random process $\{\hat z_k\}$. Let $J$ be as defined in the proof of Theorem \ref{stability_finite}.} If $J$ exists, when $J> T$,  let 
	\begin{align*}
	&\hat z_k=x_k,k<J-T,
	&\hat z_k=\varepsilon, k\ge J-T,
	\end{align*}
	where $\varepsilon$ satisfies $V(\varepsilon)=0$, when $J \le T$, let $\hat z_k=\varepsilon$ for any $k \in \N_0$; if $J$ does not exist, we let $\hat z_k=x_k$ for all $k\in\N_0$. 
	
	{If the inequality \eqref{finite_decre} is satisfied, one has  $\mathbb{E}\left[ V\left(\hat z_{k+T}\right)|\mathcal{F}_k\right]-V\left(\hat z_{k}\right)\le -\alpha V(\hat z_{k})$. Using this inequality, we next show that $V\left(\hat z_{k+T}\right)$ converges to $0$ exponentially. To this end, define a subsequence $Y^{(r)}_{m}:=V(\hat z_{mT+r}), m\in \N_0$, for each $0\le r \le T-1$. Let $\mathcal G_m^{(r)}:=\sigma(Y^{(r)}_{0},Y^{(r)}_{1},\dots,Y^{(r)}_{m})$, and one knows that $\mathcal G_m^{(r)}$ is determined if we know $\mathcal F_{mT+r}$. It then follows from the inequality \eqref{finite_decre} that for any $r$, $\mathbb E [Y_{m+1}^{r} | \mathcal G_m^{(r)} ]-Y_{m}^{(r)} \le -\alpha Y_{m}^{(r)}$. We observe from this inequality that 
	\begin{align*}
	\mathbb E \left[(1-\alpha)^{-(m+1)}Y_{m+1}^{r} | \mathcal G_m^{(r)} \right]-(1-\alpha)^{-m}Y_{m}^{(r)} \le 0.
	\end{align*}
	This means that $(1-\alpha)^{-m}Y_{m}$ is a supermartingale, and thus there is a finite random number $\bar Y^{(r)}$ such that $(1-\alpha)^{-m} Y_{m}^{r}\stackrel{a.s.}{\longrightarrow} \bar Y^{(r)}$ for any $r$. Let $\gamma= \sqrt[T]{{1/(1-\alpha)}}$, and then by  definition of $Y^{(r)}_m$ we have  $\gamma^{mT} V(\hat z_{mT+r})\stackrel{a.s.}{\longrightarrow}\bar Y^{(r)}$. Straightforwardly, $\gamma^{mT+r} V(\hat z_{mT+r})\stackrel{a.s.}{\longrightarrow} \gamma^{r} \bar Y^{(r)}$. Let $k=mT+r, \bar Y=\max_r\{\gamma^{r} \bar Y^{(r)}\}$, then it almost surely holds that $\lim_{k\to \infty} \gamma^{k} V(\hat z_{k}) \le  \bar Y$.
	From Definition \ref{def:convergence},  one concludes that $V(\hat z_{k})$ almost surely converges to $0$ exponentially no slower than $\gamma^{-1}=(1-\alpha)^{1/T}$.  From the definition of $\hat z_k$, we know that $V(\hat z_k(\omega))=V( x_k(\omega))$ for all $\omega \in \bar \Omega$, with $\bar \Omega$ defined in the proof of Theorem \ref{stability_finite}. Consequently, it holds that
	\begin{align}\label{expo_conve}
	\Pr [\lim_{k\to \infty} \gamma^{k} V(x_{k})&\le  \bar Y| \bar \Omega] \nonumber\\
	&=	\Pr [\lim_{k\to \infty} \gamma^{k} V(\hat z_{k}) \le  \bar Y| \bar \Omega]=1.
	\end{align} The proof of i) is complete since the sample paths stay in $\mathcal Q_\lambda$ with probability at least $1-V(x_0)/\lambda$. 
	
	Next, we prove ii). If the inequalities $c_1 \|x\|^a \le V(x)\le c_2 \|x\|^a$ are satisfied, and then we know that $V(x)=0$ if and only if $x=0$. Moreover, it follows from \eqref{expo_conve} that for all the sample paths that stay in $\mathcal Q_\lambda$ there holds that $ c_1 \gamma^{k} \|x\|^a \le  \gamma^{k} V(x_{k})\le  \bar Y$	since $c_1\|x_k\|^a\le V(x)$. Hence,
	$\|x_k (\omega)\|\le \left({\bar V}/{c_1}\right)^{1/a}\gamma^{-k/a}$ for any $\omega \in \bar \Omega$, and one can check that this inequality holds with probability at least $1-V(x_0)/\lambda$. If $x_0\to 0$, we know that $1-V(x_0)/\lambda\to 1$, which completes the proof. }
\end{proof}

If $\mathcal Q_\lambda$ is positively invariant, the following corollary follows straightforwardly.
{
\begin{corollary}\label{Coroll_invariant_expo}
	If $\mathcal Q_\lambda$ is positively invariant w.r.t the system \eqref{model:main} and suppose {the assumptions a) and b)} of Theorem \ref{stability_finite} are satisfied with the inequality of b) strengthened to \eqref{finite_decre}, the following statements apply:
	
	\emph{i)} for any given $x_0 \in \mathcal Q_\lambda$, $V(x_k)$ converges to $0$ exponentially  {no slower than} $(1-\alpha)^{{1}/{T}}$ with probability one;

	\emph{ii)} if moreover $V$ satisfies that $c_1 \|x\|^a \le V(x)\le c_2 \|x\|^a$ for some $c_1,c_2, a >0$, then $x=0$ is locally a.s. exponentially stable in $\mathcal Q_\lambda$. Furthermore, if $\mathcal Q_\lambda=\mathbb R^n$, then $x=0$ is globally a.s. exponentially stable.
\end{corollary} 

{The following corollary, which can be proven following the same lines as Theorems \ref{stability_finite} and \ref{stability_expone_finite}, shares some similarities to LaSalle's theorem for deterministic systems. It is worth mentioning that the function $V$ here does not have to be radially unbounded.
	\begin{corollary}\label{lassale-like}
		Let $\mathbb D\subset \mathbb R^n$ be a compact set that is positively invariant w.r.t the system \eqref{model:main}. Let $V:\mathbb{R}^n\to \mathbb{R}$ be a {continuous nonnegative function}, and $\bar {\mathcal Q}_\lambda:=\{x\in \mathbb D: V(x)< \lambda \}$ for some $\lambda>0$. Assume that
		$\mathbb{E}\left[ V\left(x_{k+1}\right)|\mathcal F_k\right]-V\left(x_k\right)\le 0$ for all $k$ such that $x_k \in \bar {\mathcal {Q}}_\lambda$, then 
		
		\emph{i)}  if there is an integer $T \ge 1$, independent of $\omega$, such that for any $k \in \N_0$, $\mathbb{E}\left[ V\left(x_{k+T}\right)|\mathcal{F}_k\right]-V\left(x_k\right)\le -\varphi(x_k)$, where $\varphi:\mathbb{R}^n\to \mathbb R$ is continuous and satisfies $\varphi(x)\ge 0$ for any $x \in \bar{\mathcal Q}_\lambda$, then for any initial condition $x_0 \in \bar{\mathcal Q}_\lambda$, $x_k$ converges to $\bar{\mathcal{D}}_1:=\{x \in \bar{\mathcal Q}_{\lambda}:\varphi(x)=0\}$ with probability  at least $1-V(x_0)/\lambda$;
		
		\emph{ii)} if the inequality in a) is strengthened to $\mathbb{E}\left[ V\left(x_{k+T}\right)|\mathcal{F}_k\right]$ $-V\left(x_k\right)\le -\alpha V(x_k)$ for some $0<\alpha <1$, then for any given $x_0 \in \bar{\mathcal Q}_\lambda$, $V(x_k)$ converges to $0$ exponentially at a rate no slower than $(1-\alpha)^{{1}/{T}}$,  and $x_k$  converges to $\bar {\mathcal D} _2:=\{x\in \bar{\mathcal Q}_{\lambda}:V(x)=0\}$, with probability at least $1-V(x_0)/\lambda$;
		
		\emph{iii)} if $\bar {\mathcal Q}_\lambda$ is positively invariant w.r.t the system \eqref{model:main}, then all the convergence in both i) and ii) takes place \emph{almost surely}. 
 	\end{corollary}
}
}

\textit{Example 1 Cont.} Now let us look back at Example 1 and still choose $V(x)=\left\| x\right\|_\infty$ as a stochastic Lyapunov function candidate. It is easy to see that $V(x)$ is a nonnegative supermartingale. To show the stochastic convergence, let $T=2$ and one can calculate the conditional expectations 
\begin{align*}
&\mathbb{E}\left[ {\left. {V\left( {x_{k + T}} \right)} \right|x_k,y_k=1} \right] - V\left( {x_k} \right)\\
&= {0.5}{\left\| {\begin{array}{*{20}{c}}
		{0.2x_k^1}\\
		{0.8x_k^2}
		\end{array}} \right\|_\infty } + {0.5}{\left\| {\begin{array}{*{20}{c}}
		{0.2x_k^1}\\
		{0.6x_k^2}
		\end{array}} \right\|_\infty } - {\left\| {\begin{array}{*{20}{c}}
		{x_k^1}\\
		{x_k^2}
		\end{array}} \right\|_\infty }\\
&\le  - 0.3V\left( {x_k} \right), \forall x_k \in \mathbb R^2.
\end{align*}
When $y_k=2,3$, there analogously hold that
\begin{align*}
	\mathbb{E}[ \left. {V\left( {x_{k + T}} \right)} \right|x_k, y_k] - V\left( {x_k} \right) 
	\le -0.3V(x_k), \forall x_k \in \mathbb R^2.
\end{align*}
From these three inequalities one can observe that starting from any initial condition $x_0$, $\mathbb{E}V(x)$ decreases at an exponential speed after every two steps before it reaches $0$. By Corollary \ref{Coroll_invariant_expo}, one knows that origin is {globally a.s. exponentially stable}, consistent with our conjecture. \hfill   $ \Box $ 

\begin{remark}
 Kushner and other researchers  have used more restricted conditions to construct Lyapunov functions than those appearing in our results to analyze asymptotic or exponential stability of random processes \cite{Kushner_1971,Kushner_1965,beutler1973two}. It is required that $\mathbb{E}[V(x_k)]$ decreases strictly at every  step, until $V(x_k)$ reaches a limit value. However, in our result, this requirement is relaxed.  In addition, Kushner's results rely on the assumption that the underlying random process is Markovian, but we work with more general random processes. 
\end{remark}

In the following sections, we will show how the new Lyapunov criteria can be applied  to distributed computation.

\section{Products of Random Sequences of Stochastic Matrices}\label{sec_pro}

In this section, we study the convergence of products of stochastic matrices, where the obtained results on finite-step Lyapunov functions are used for analysis. {Let $\Omega_0:=\{1,2,\dots,m\}$ be the state space and $\mathcal M:=\{F_1,F_2,\dots,F_m\}$ be the set of $m$ stochastic matrices $F_i\in\R^{n\times n}$. } Consider a random sequence $\{W_\omega(k):k\in\N\}$ {on the probability space $(\Omega,\mathcal{F},\Pr)$, where $\Omega$ is the collection of all infinite sequences $\omega=(\omega_1,\omega_2,\dots)$ with $\omega_k\in \Omega_0$, and we define $W_\omega(k):=F_{\omega_k}$. For notational simplicity, we denote $W_\omega(k)$ by $W(k)$.} For the backward product of stochastic matrices
\begin{equation}\label{product}
{W(t+k,t)}=W({t+k})\cdots W({t+1}),
\end{equation}
where $k \in \N,t \in \N_0$, we are interested in establishing conditions on  $\{W(k)\}$, under which there holds that $\lim_{k \to \infty}W(k,0)=L$
for a random matrix $L=\mathbf{1}\xi^\top$ where $\xi \in \mathbb{R}^n$ satisfies $\xi^\top\mathbf{1}=1$. 

Before proceeding, let us introduce some concepts in probability.  {Let $\mathcal F_k =\sigma(W(1),\dots,W(k))$, so that evidently $\{\mathcal F_k\}$, $k=1,2,\dots,$ is an increasing sequence of $\sigma$-fields.}
 Let $\phi: \Omega\to\Omega$ be the shift operator, i.e., $\phi(\omega_1,\omega_2,\dots)=(\omega_2,\omega_3,\dots)$. A random sequence of stochastic matrices $\{W(1),W(2),\dots,W(k),\dots\}$ is said to be \textit{stationary} if the shift operator is measure-preserving. In other words, the sequences $\{W({k_1}),W({k_2}),\dots,W(k_r)\}$ and $\{W({k_1+\tau}),W({k_2+\tau}),\dots,W({k_r+\tau})\}$ have the same joint distribution for all $k_1, k_2,\dots,k_r$ and $\tau\in \N$. Moreover, a sequence is said to be \textit{stationary ergodic} if it is stationary, and every invariant set $\mathcal{B}$ is trivial, i.e., for every $A\in \mathcal{B}$, $\Pr[A] \in \{0,1\}$. Here by a invariant set $\mathcal{B}$, we mean $\phi^{-1}\mathcal{B}=\mathcal{B}$.


\subsection{Convergence Results}\label{convergence:matrices}

We first introduce three classes of stochastic matrices,  denoted by $\mathcal{M}_1,\mathcal{M}_2$, and $\mathcal{M}_3$, respectively. We say $A\in \mathcal{M}_1$ if $A$ is indecomposable, and aperiodic (such stochastic
matrices are also referred to as SIA for short); $A\in \mathcal{M}_2$ if $A$ is scrambling, i.e., no two rows of $A$ are orthogonal; and $A \in \mathcal{M}_3$ if $A$ is Markov, i.e., there exists a column of $A$ such that all entries in this column are positive \cite[Ch. 4]{Nonnegative_Matrices}. 

Coefficients of ergodicity serve as a fundamental tool in analyzing the convergence of products of stochastic matrices. In this paper, we employ a standard one. For a stochastic matrix $A\in \mathbb{R}^{n\times n}$, the coefficient of ergodicity $\tau(A)$ is defined by
\begin{equation}\label{coeffi_ergo}
\tau \left( A \right) = 1-\min \limits_{i,j}\sum\limits_{s=1}^n \min(a_{is},a_{js}).
\end{equation}
It is known that this coefficient of ergodicity satisfies $0\le \tau(A) \le1$, and $\tau(A)$ is proper since $\tau(A)=0$ if and only if all the rows of $A$ are identical. Importantly, it holds that
\begin{equation}\label{coefficients_scrambling}
\tau(A)< 1
\end{equation}
if and only if $A \in \mathcal{M}_2$ (see {\cite[p.82]{Nonnegative_Matrices}}). For any two stochastic matrices $A,B$, the following property will be critical for the proof in Appendix \ref{appendix_1}:
\begin{equation}\label{coeffi_submulti}
\tau(AB)\le \tau(A)\tau(B).
\end{equation}

To proceed, we make the following assumption for the sequence $\{W(k)\}$.

\begin{assumption}\label{assumption_1}
	Suppose the sequence of stochastic matrices $\{W(k)\}$ is driven by a random process satisfying the following conditions.
	\begin{enumerate}
		\item[a)] There exists an integer $h>0$ such that 
		\begin{equation}\label{assumption_re}
		\Pr \left[ W(k+h,k)\in \mathcal{M}_2 \right]>0	
		\end{equation}
		holds for any $k \in \N_0$, and 
			\begin{equation}\label{pro_to_inf_gener}
						\sum\limits_{i = 1}^\infty  {\Pr \left[ {W\big( {k + ih,k + \left( {i - 1} \right)h} \big)} \in \mathcal{M}_2 \right]}  = \infty, \forall k. 
			\end{equation}			
		\item[b)] There is a positive number $\alpha$ such that $W_{ij}(k)\ge \alpha$ whenever $W_{ij}(k)>0$.
	\end{enumerate}
\end{assumption}

Now we are ready to provide our main result on the convergence of stochastic matrices' products.

\begin{theorem}\label{theorem_geral}
	Under Assumption \ref{assumption_1}, the product of the random sequence of stochastic matrices  $W(k,0)$ converges to a random matrix $L=\mathbf{1}\xi^\top$  almost surely as $k\to \infty$.	
\end{theorem}

To prove Theorem \ref{theorem_geral}, consider the stochastic discrete-time dynamical system described by
\begin{equation}\label{system_dyn}
x_{k+1}=W_{y(k+1)}x_{k}:=W(k+1)x_{k},
\end{equation}
for all $k\in\N_0$, where $x_k\in\mathbb{R}^n$,  the initial state $x_0$ is a constant with probability one, ${y(k)}$ is regarded as randomly switching signal, and $\{W(1),W(2),\dots\}$ is the random process of stochastic matrices we are interested in. One knows that $x_k$ is adapted to $\mathcal{F}_k$. Thus, to investigate the limiting behavior of the product \eqref{product}, it is sufficient to study the limiting behavior of system dynamics \eqref{system_dyn}. We say the state of system \eqref{system_dyn} reaches an \textit{agreement} state if $\lim_{k \to \infty}x_k=\mathbf{1} \xi$ for some $\xi \in \mathbb{R}$. Then the agreement of system \eqref{system_dyn} for any initial state $x_0$ implies that $W(k,0)$ converges to a rank-one matrix as $k \to \infty$  \cite{Touri}.

To investigate the agreement problem, we define 
$\left\lceil {x_k} \right\rceil:=\max_{1 \le i \le n}x^i_k,\left\lfloor {x_k} \right\rfloor:=\min_{1 \le i \le n}x^i_k$, and 
\begin{equation}\label{eq:Lyapunov}
v_k=\left\lceil {x_k} \right\rceil-\left\lfloor {x_k} \right\rfloor.
\end{equation}
 For any $k\in\N$, $v_k$ is adapted to  $\mathcal{F}_k$ since $x_k$ is. {The agreement is said to be reached asymptotically almost surely if $v_k \stackrel{a.s.}{\longrightarrow} 0$ as $k\to \infty$, and it is said to be reached exponentially almost surely with convergence rate no slower than $\gamma^{-1}$ if there exists $\gamma>1$ such that $\gamma^k v_k \stackrel{a.s.}{\longrightarrow} y$ for some finite $y\ge 0$.} The random variable $v_k$ has some important properties given by the following proposition.

\begin{proposition}\label{inequality_v}
	Let $x_{k+1}=Ax_k$, where $A$ is a stochastic matrix. Then $v_{k+1} \le v_k$, and $v_{k+1} < v_{k}$ for any $x_{k} \notin {\rm span} (\mathbf{1})$ \textit{if and only if} $A$ is scrambling (i.e., $A\in \mathcal{M}_2)$. 
\end{proposition}

\begin{proof}
	It is shown in \cite{Nonnegative_Matrices} that $v_{k+1}\le \tau{(A)} v_{k}$ with $\tau(\cdot)$ defined in \eqref{coeffi_ergo}. Therefore, the sufficiency follows from \eqref{coefficients_scrambling} straightforwardly. We then prove the necessity by contradiction. Suppose $A$ is not scrambling, and then there must exist at least two rows, denoted by $i,j$, that are orthogonal. Define the two sets $\mathbf{i}:=\{l:a_{il} > 0,l\in \mathbf N\}$ and $\mathbf{j}:=\{m:a_{jm} > 0,m \in \mathbf N\}$, respectively. It follows then from the scrambling property that $\mathbf{i} \cap \mathbf{j}=\emptyset$.	
	Let $x^q_{k}=1$ for all $q\in \mathbf{i}$, $x^q_{k}=0$ for all $q\in \mathbf{j}$, and let $x^m_{k}$ be any arbitrary positive number less than 1 for all $m\in \mathbf N \backslash (\mathbf{i} \cup \mathbf{j})$ if $\mathbf N \backslash (\mathbf{i} \cup \mathbf{j})$ is not empty. Then the states at time $k+1$ become
	\begin{align*}
	&x^i_{k + 1}  = \sum\limits_{l = 1}^n {a_{il}}{x^l_k}  = \sum\limits_{l \in \mathbf{i}} {a_{il}}{x^l_k}  = 1, \\
	&x^j_{k + 1}  = \sum\limits_{l = 1}^n {{a_{jl}}{x^l_k}}  = \sum\limits_{l \in \mathbf{j}} {{a_{jl}}{x^l_k}}  = 0,
	\end{align*}
	and $0\le x^m_{k+1}\le 1$ for all $m\in \mathbf N \backslash (i \cup j)$.
	This results in $v_{k+1}=v_k=1$. By contradiction one knows that a scrambling $A$ is necessary for $v_{k+1} < v_k$, which completes the proof.
\end{proof}

In order to prove Theorem \ref{theorem_geral}, the following intermediate result is useful.

\begin{proposition}\label{scrambling_bounded}
	For any scrambling matrix $A\in \mathbb{R}^{n\times n}$, the coefficient of ergodicity $\tau(A)$ defined in \eqref{coeffi_ergo} satisfies
	\begin{equation}
	\tau(A)\le 1-\gamma 
	\nonumber
	\end{equation}
	if all the positive elements of $A$ are lower bounded by $\gamma>0$.
\end{proposition}
\begin{IEEEproof}
	Consider any two rows of $A$, denoted by $i,j$. Define two sets, $\mathbf{i}:=\{s:a_{is}>0\}$ and $\mathbf{j}:=\{s:a_{js}>0\}$. From the scrambling hypothesis, one knows that $\mathbf{i} \cap \mathbf{j} \ne \emptyset$. Thus it holds that
	\[\sum\limits_{s = 1}^n {\min \left( {{a_{is}},{a_{js}}} \right)}  = \sum\limits_{s \in \mathbf{i} \cap \mathbf{j}} {\min \left( {{a_{is}},{a_{js}}} \right) \ge \gamma }. \]
	Then from the definition of $\tau(A)$, it is easy to see 
	\[\tau \left( A \right) = 1 - \mathop {\min }\limits_{i,j} \sum\limits_{s = 1}^n {\min } \left( {{a_{is}},{a_{js}}} \right) \le 1 - \gamma, \]
	which completes the proof.
\end{IEEEproof}
We are in the position to prove Theorem \ref{theorem_geral} by showing that $v_k \stackrel{a.s.}{\longrightarrow} 0$ as $k\to\infty$, where Theorem \ref{stability_finite} and Corollary \ref{Coroll_invariant} will be used.

\begin{IEEEproof}[Proof of Theorem \ref{theorem_geral}]
	Let $V(x_k)=v_k$ be a finite-step stochastic Lyapunov function candidate for the system dynamics \eqref{system_dyn}. It is easy to see $V(x)=0$ if and only if $x\in {\rm span}(\mathbf{1})$. Since all $W(k)$ are stochastic matrices, we observe that $\mathbb{E}[V(x_{k+1})|{\mathcal{F}_k}]-V(x_k)\le 0$ from Proposition \ref{inequality_v}, which implies that $V(x_k)$ is exactly a supermartingale with respect to $\mathcal F_k$. From Lemma \ref{supermartingale}, we know $V(x_k)\stackrel{a.s.}{\longrightarrow} {\bar V}$ for some ${\bar V}$ because $V(x_k)\ge0$ and $\mathbb{E}V(x_k)<\infty$. From Assumption \ref{assumption_1}, we know that there is an $h$ such that the product $W(k+h,k)$ is scrambling with positive probability for any $k$. Let $\mathcal{W}_k$ be the set of all possible $W(k+h,k)$ at time $k$, and  $n_k$  the cardinality of $\mathcal{W}_k$. Let $n_k^s$ be the number of scrambling matrices in $\mathcal{W}_k$. We denote each of these scrambling matrices and each of non-scrambling matrices by $S_k^i,i=1,\dots,n_k^s$ and $\bar{S}_k^j,j=1,\dots,n_k-n_k^s$, respectively. The probabilities of all the possible $W(k+h,k)$ sum to 1, i.e., 
	\begin{equation}\label{proba_sum}
	\sum\limits_{i = 1}^{n_k^s} {\Pr \left[ {{S^i_k}} \right]}  + \sum\limits_{j = 1}^{{n_k} - n_k^s} {\Pr \left[ \bar{S}_k^j \right]}  = 1.
	\end{equation}
	Then the conditional expectation of $V(x)$ after finite steps for any $k$ becomes
	\begin{align*}
	\mathbb{E}\left[ {\left. {V\left( {x_{k + h} } \right)} \right|{\mathcal{F}_k}} \right] &- V\left( {x_ k} \right)\\
	&= \mathbb{E}\left[ {V\big( {W\left( {k+h,k} \right)x_k} \big)} \right]- V\left( {x_k} \right)\\
	&\le \mathbb{E} \left[\tau \big( {W\left( { k+h,k} \right)} \big)\right] V\big( {x_k} \big) - V\left( {x_k} \right),
	\end{align*}	
	where $\tau(\cdot)$ is given by \eqref{coeffi_ergo}. One can calculate that 
	\begin{align*}
		\mathbb{E}\left[ {\tau \Big(W\left( {k+h,k} \right)\Big)} \right] &- 1\\
	= \sum\limits_{i = 1}^{n_k^s} {\Pr \left[ {S_k^i} \right]} \tau \left( {S_k^i} \right) &+ \sum\limits_{j = 1}^{{n_k} - n_k^s} {\Pr \left[ {\bar S_k^j} \right]} \tau \left( {\bar S_k^j} \right) - 1\\
	&\le \sum\limits_{i = 1}^{n_k^s} {\Pr \left[ {S_k^i} \right]} \Big( {\tau \left( {S_k^i} \right) - 1} \Big),
	\end{align*}
	where Proposition \ref{inequality_v} and equation \eqref{proba_sum} have been used. From Assumption \ref{assumption_1}.b), we know that the positive elements of $W(k)$ are lower bounded by $\alpha$, and thus  the positive elements of $S_k^i$ in \eqref{lyapunov_decrease}  are lower bounded by $\alpha^h$. Thus $\tau(S_k^i)\le1-\alpha^h$ according to Proposition \ref{scrambling_bounded}, and it follows that
	\begin{align}
	\mathbb{E}[ \left. {V\left(x_ {k + h}\right)} \right|&{{\cal F}_k} ]- V\left( {x_k} \right) \nonumber\\
	&\le -\sum\limits_{i = 1}^{n_k^s} {\Pr \left[ {S_k^i} \right]} \alpha^h \mathbb E  V(x_k): = {\varphi _k}\left(x_k\right).\label{lyapunov_decrease}
	\end{align}
	By iterating, one can easily show that
	\begin{align}
	\mathbb{E}\left[ {{V\left(x_ {nh}\right)}} \right]&- V\left( {x_0} \right)\le -\sum\nolimits_{k=0}^{n-1}{\varphi _k}\left(x_k\right) \nonumber\\
	&=-\sum\nolimits_{k = 0}^{n-1}\sum\nolimits_{i = 1}^{n_k^s} {\Pr \left[ {S_k^i} \right]} \alpha^h \mathbb E  V(x_k).
	\end{align}
	It then follows that $V\left( {x_ 0 }\right)-\mathbb{E}\left[ {{V\left(x_{nh} \right)}} \right] <\infty$ even when $n \to \infty$, since $V(x)\ge 0$. According to the condition  \eqref{pro_to_inf_gener}, we know $\sum_{k = 0}^{n-1}\sum_{i = 1}^{n_k^s} {\Pr \left[ {S_k^i} \right]}=\infty$. By contradiction, it is easy to infer that $\mathbb E  V(x_k)\stackrel{a.s.}{\longrightarrow} 0$. Since we have already shown that $V(x_k)\stackrel{a.s.}{\longrightarrow} \bar V$ for some random $\bar V \ge0$, one can conclude that $V(x_k)\stackrel{a.s.}{\longrightarrow} 0$.
For any given $x_0 \in \mathbb{R}^n$, define the compact set $\mathcal Q:=\{x:\left\lceil {x} \right\rceil\le\left\lceil {x_0} \right\rceil,\left\lfloor {x} \right\rfloor\ge\left\lfloor {x_0} \right\rfloor$. For any random sequence $\{W(k)\}$, it follows from the system dynamics \eqref{system_dyn} that
	\begin{equation}	
	\begin{array}{l}
	\left\lceil{x_k}\right\rceil \le \left\lceil{x_{k-1}}\right\rceil \le\cdots\le\left\lceil{x_1}\right\rceil\le \left\lceil{x_0}\right\rceil,\\
	\left\lfloor{x}_k\right\rfloor \ge \left\lfloor{x}_{k-1}\right\rfloor \ge \cdots\ge \left\lfloor x_1\right\rfloor \ge \left\lfloor{x}_0\right\rfloor,
	\end{array}
	\nonumber
	\end{equation}	
	and thus $x_k$ will remain within $\mathcal Q$. From {Corollary \ref{lassale-like}}, we know that $x_k$ asymptotically converges to $\{x \in \mathcal Q :\varphi_k(x)=0\}$, or equivalently, $\{x\in \mathcal Q:V(x)=0\}$ {almost surely} as $k\to \infty$  since $V(x)$ is continuous. In other words, for any $x_0\in\mathbb{R}^n$, $x_k \stackrel{a.s.}{\longrightarrow} \zeta \mathbf{1}$ for some $\zeta \in \mathbb{R}$, which proves Theorem \ref{theorem_geral}.	
\end{IEEEproof}

For a random sequence of stochastic matrices, Theorem \ref{theorem_geral} has provided a quite relaxed condition for the backward product \eqref{product} determined by the random sequence $\{W(k)\}$ to converge to a rank-one matrix: over any time interval of length $h$, i.e., $[h+t,t]$ for any $t \ge 0$, the product $W(t+h)\cdots W({t+1})$ has positive probability to be scrambling. The following corollary follows straightforwardly since any Markov matrix is certainly scrambling.

\begin{corollary}\label{scrambling_cora}
	For a random sequence $\{W_k\}_{k=1}^\infty$, the product \eqref{product} converges to a random matrix $L=\mathbf{1}\xi^\top$  almost surely if there exists an integer $h$ such that $W(t+h,t)$ becomes a Markov matrix for any $k$ with positive probability and $\sum\nolimits_{i = 1}^\infty  {\Pr \left[ {W\left( {k + ih,k + \left( {i - 1} \right)h} \right)} \in \mathcal{M}_3 \right]}  = \infty, \forall k$.
\end{corollary}

Next we assume that the sequence $\{W(k)\}$ is driven by an underlying \textit{stationary} process. Then the condition in Theorem \ref{theorem_geral} can be further relaxed. 
 
 \begin{assumption}\label{assumption_2}
 	Suppose the random sequence of stochastic matrices $\{W(k)\}$ is driven by a stationary process satisfying the following conditions.
 	\begin{enumerate}
 		\item[a)] There exists an integer $h>0$ such that 
 		\begin{equation}\label{assumption_re2}
 		\Pr \left[ W(k+h,k)\in \mathcal{M}_1 \right]>0	
 		\end{equation}
 		holds for any $k \in \N_0$. 		
 		\item[b)] There is a positive number $\alpha$ such that $W_{ij}(k)\ge \alpha$ whenever $W_{ij}(k)>0$.
 		\end{enumerate}
 \end{assumption}
In other words, Assumption \ref{assumption_2} suggests that any corresponding matrix product of length $h$ becomes an SIA matrix with positive probability, and the positive elements for all $W(k)$ are uniformly lower bounded away from some positive value.
 
\begin{theorem}\label{theorem_stationary}
		Under Assumption \ref{assumption_2}, the product of the random sequence of stochastic matrices  $W(k,0)$ converges to a random matrix $L=\mathbf{1}\xi^\top$  almost surely.	
\end{theorem}

If two stochastic matrices $A_1$ and $A_2$ have zero elements in the same positions, we say these two matrices are of the same type, denoted by $A_1\sim A_2$. Obviously, there holds the trivial case  $A_1 \sim A_1$. One knows that for any SIA matrix $A$, there exists an integer $l$ such that $A^l$ is scrambling; it is easy to extend this to the inhomogeneous case, i.e., any product of $l$ stochastic matrices of the same type of $A$ is scrambling if all the matrices are element-wise lower bounded.

\begin{IEEEproof}[Proof of Theorem \ref{theorem_stationary}]
Since $\{W(k)\}$ is driven by a stationary process, we know that $\{W \left( t+h\right),\dots, W \left(t+1 \right) \}$ has the same joint distribution as $\{W \left( t+2h\right),\dots,$ $W \left(t+h+1 \right) \}$  for any $t \in \N_0,h \in \N$. For the $h$ given in Assumption \ref{assumption_2}, there exists an SIA matrix $A$ such that $\Pr [W\big(t+kh+h,t+kh+1\big)=A]>0$. Thus it follows that $\Pr [W\big(t+kh+2h,t+kh+1\big)=A]>0$ for any $k\in \N_0$. Thus
	\[{\Pr \left[ {\left. {\begin{array}{*{20}{c}}
				{W\big({t} + (k + 2)h,{t} + (k + 1)h\big)}\\
				{\sim W\big({t} + (k + 1)h,{t} + kh\big)}
				\end{array}} \right|W\left( {h,{t} + kh} \right)} \right] > 0}.\]
	When $W(t+h,t)\in \mathcal{M}_1$, which happens with positive   probability, we have
	\begin{align*}	
	\Pr& \left[ {\begin{array}{*{20}{l}}
			{W(t + 2h,t + h)\sim W(t + h,t),}\\
			{\;\;\;\;\;\;\;\;\;\;\;\;\;\;\;\;\;\;\;\;\;\;\;\;\;\;\;\;\;\;\;\;\;\;W(t + h,t) \in {{\cal M}_1}}
			\end{array}} \right]\\
	 &= \Pr \left[ {\left. \begin{array}{l}
			W(t + 2h,t + h)\\
			\;\;\;\;\;\;\;\;\sim W(t + h,t)
			\end{array} \right|\Pr \left[ {W(t + h,t) \in {{\cal M}_1}} \right]} \right]\\
	&\;\;\;\;\;\;\;\;\;\;\;\;\;\;\;\;\;\;\;\;\;\; \;\;\;\;\;\;\;\;\;\;\;\;\;\;\; \cdot \Pr \left[ {W(t + h,t) \in {{\cal M}_1}} \right] > 0.
	\end{align*}
	 By  recursion one can conclude that all the $m$ products $W({t} + (k+1)h,{t} + kh), k \in \{0,\dots,m-1\}$, occur as the same SIA type with positive probability.
	Since all the products $W({t} + (k+1)h,{t} + kh)$ are of the same type, one can  choose  $m$ such that $W(t +m h,t)$ is scrambling. This in turn implies that $\Pr \left[ W(t +mh,t)\in \mathcal{M}_2 \right]>0$, and the property of stationary process makes sure that \eqref{pro_to_inf_gener} holds. The conditions in Assumption \ref{assumption_1} are therefore all satisfied, and then Theorem \ref{theorem_stationary} follows from Theorem \ref{theorem_geral}.
	\end{IEEEproof}

\begin{remark}Theorems \ref{theorem_geral} and  \ref{theorem_stationary} have established some sufficient conditions for the convergence of a random sequence of stochastic matrices to a rank-one matrix. A further question is how these results can be applied to control distributed computation processes. To answer this question, let us consider a finite set of stochastic matrices $\mathcal L=\{F_1\dots,F_m\}$, from which each $W(k)$ in the random sequence $\{W(k)\}$ is sampled. 
It is defined in \cite{blondel2014decide} that $\mathcal{L}$ is a consensus set if the arbitrary product $\prod_{i=1}^{k}W(i), W(i)\in\mathcal{L}$, converges to a rank-one matrix. However, it has also been shown that to decide whether $\mathcal{L}$ is a consensus set is an NP-hard problem \cite{blondel2014decide,xia_2}. For a non-consensus set $\mathcal{L}$, it is always not obvious how to find a deterministic sequence that converges, especially when $\mathcal{L}$ has a large number of elements and $F_i$ has zero diagonal entries. However, the convergence can be ensured almost surely by introducing some randomness in the sequence, provided that there is a convergent deterministic sequence intrinsically.
\end{remark}


\subsection{Estimation of Convergence Rate}
In Section \ref{convergence:matrices}, we have shown how the product $W(k,0)$ determined by a random process asymptotically converges to a  rank-one matrix $W$ a.s. as $k\to\infty$. However, the convergence rate for such a randomized product is not yet clear. It is quite challenging to investigate how fast the process converges, especially when each $W(k)$ may have zero diagonal entries. In this subsection, we address this problem by employing finite-step stochastic Lyapunov functions.
Now let us present the main result on the convergence rate.

\begin{theorem}\label{convergence_rate_pro_bounded}
	In addition to Assumption \ref{assumption_1}, if there exist a number $p$, $0<p<1$, such that
	\[\Pr \left[ W(h,t)\in \mathcal{M}_2 \right] \ge p>0,\]
	 then the almost sure convergence of the product of $W(k,0)$ to a random matrix $L=\mathbf{1}\xi^\top$ is exponential, and the rate is {no slower than}  $\left( {1 - p{\alpha ^h}} \right)^{1/h}$.
\end{theorem}

\begin{IEEEproof}
	Choosing $V\left(x_k\right)=v_k$ as a finite-step stochastic Lyapunov function candidate, from \eqref{lyapunov_decrease} we have 
	\begin{align}\label{lyapunov_decrease_2}
	\mathbb{E}\left[ {\left. {V\left( {x_ {k + h}} \right)} \right|{\mathcal{F}_k}} \right] &- V\left( {x_k} \right) \nonumber  \\
	&\le-\sum\limits_{i = 1}^{n_k^s} {\Pr \left[ {S_k^i} \right]}  \alpha^h V\left( {x_k} \right).
	\end{align}
	Furthermore, it is easy to see that 
	\[\sum\limits_{i = 1}^{n_k^s} {\Pr \left[ {S_k^i} \right]}=\Pr \left[ W(h,t)\in \mathcal{M}_2 \right] \ge p,\]
	Substituting it into \eqref{lyapunov_decrease_2} yields 
	\begin{equation}\label{lyapunov_decrease_exponentially}	
	\mathbb{E}\left[ {\left. {V\left( {x_{k + h}} \right)} \right|{\mathcal{F}_k}} \right]
	\le \left(1-p  \alpha^h \right) V\left( {x_k} \right). \nonumber
	\end{equation}
	It follows from {Corollary \ref{lassale-like}} that ${V\left( {x_{k + h}} \right)}\stackrel{a.s.}{\longrightarrow} 0$, with an convergence rate {no slower than} $\left( {1 - p{\alpha ^h}} \right)^{1/h}$. In other words, the agreement is reached exponentially almost surely, which implies Theorem \ref{convergence_rate_pro_bounded}.
\end{IEEEproof}

Theorem \ref{convergence_rate_pro_bounded} has established the almost sure exponential convergence rate for the product of  $\{W(k)\}$. If any subsequence $\{W(k+1), \dots,W(k+2),W(k+h)\}$ can result in a scrambling product $W(k+h,k)$ with positive probability and this probability is lower bounded away by some positive number, and then the convergence rate is exponential. \textit{Interestingly, the greater this lower bound is, the faster the convergence becomes.} If we consider a special random sequence which is driven by a stationary ergodic process, the exponential convergence rate follows without any other conditions apart from Assumption \ref{assumption_2}, and an alternative proof is given in Appendix \ref{appendix_1}.

\begin{corollary}\label{stationary_ergodic}
	If the random process governing the evolution of the sequence $\{W(k)\}$ is \emph{stationary ergodic}, the product $W(k,0)$ converges to a random rank-one matrix at an exponential rate almost surely if the conditions of Assumption \ref{assumption_2} are satisfied.
\end{corollary}

\subsection{Connection to Markov Chains}
In this subsection, we show that Theorems \ref{theorem_stationary}, and \ref{convergence_rate_pro_bounded} are the generalizations of some well known results for Markov chains in \cite{Wolfowitz,Nonnegative_Matrices}. A fundamental result on inhomogeneous Markov chains is as follows. 
\begin{lemma}[{\cite[Th. 4.10]{Nonnegative_Matrices}}, \cite{Wolfowitz}]
	If the product $W(k,t)$, formed from a sequence $\{W(k)\}$, satisfies $W(t+k,t)\in \mathcal{M}_1$ for any $k\ge 1,t\ge 0$, and $W_{ij}(k)\ge \alpha$ whenever $W_{ij}(k)>0$, then $W(k,0)$ converges to a rank-one matrix.
\end{lemma}
Let $h$ be the number of distinct types of scrambling matrices of order $n$. It is known that the product $W(t+h,t)$ is scrambling for any $t$. In this case, we may take the probability of each product $W(t+h,t)$ being scrambling as $p=1$, and as an immediate consequence of Theorem \ref{convergence_rate_pro_bounded}, we know that $W(k,0)$ converges to a rank-one matrix at a exponential rate that is {no slower than} $(1-\alpha^h)^{{1}/{h}}$. This convergence rate is consistent with what is estimated in {\cite[Th. 4.10]{Nonnegative_Matrices}}. This also applies to the homogeneous case where $W(k)=W_1$ for any $k$ with $W_1$ being scrambling. Moreover, it is known that the condition can be relaxed by just requiring $W_1$ to be SIA to ensure the convergence, which is an immediate consequence of Theorem \ref{theorem_stationary}.

In next section, we discuss how the results can be further applied to the context of asynchronous computations.

\section{Asynchronous Agreement over Possibly Periodic Networks}
In this section, we take each component $x^j$ in $x$ from \eqref{system_dyn} as the state of agent $i$ in an $n$-agent system. Define the distributed coordination algorithm 
\begin{equation}\label{eq:syst_in}
{x^i}\left( t_{k + 1} \right) = \sum\limits_{j = 1}^n {{w_{ij}}x^j\left(t_k \right)}, k \in \N_{ 0},i \in \mathbf{N},
\end{equation}
where the averaging weights $w_{ij} \ge 0$, $\sum_{j=1}^{n}w_{ij}=1$, and $t_k$ denote the time instants when updating actions happen. Here we assume the initial state $x(t_0)$ is given. It is always assumed that $T_1\le t_{k+1}-t_k\le T_2$, where $t_0=0$ and $T_1,T_2$ are positive numbers.  We say the states of system \eqref{eq:syst_in} reach \textit{agreement} if $\lim_{k \to \infty}x(t_k)=\mathbf{1} \zeta $, mentioned in Section \ref{sec_pro}.
Let $W=[w_{ij}] \in \mathbb{R}^{n\times n}$, and obviously $W$ is a stochastic matrix. The algorithm \eqref{eq:syst_in} can be rewritten as $x(t_{k+1})=Wx(t_k)$. In fact, the matrix $W$ can be associated with a directed, weighted graph $\mathcal{G}_W=\left( \mathcal{V,E} \right)$, where $\mathcal{V}:=\{1,2,\cdots,n\}$ is the vertex set and $\mathcal{E}$ is the edge set for which $(i,j) \in \mathcal{E}$ if $w_{ji}>0$. The graph $\mathcal{G}_W$ is called a \textit{rooted} one if there exists at least one vertex, called a \textit{root}, from which any other vertex can be reached. It is known that agents are able to reach agreement for all $x(0)$ if  $W$ is SIA (\cite{Wolfowitz,Nonnegative_Matrices}). {However, the situations when  $W$ is not SIA have not been studied before, although they appear often in real systems, such as social networks. As we are interested in studying the agreement problem when $W$ is possibly periodic, let us define periodic stochastic matrices.} 
\begin{definition} \label{periodic}
	A stochastic matrix $A \in  \mathbb{R}^{n \times n}$ is said to be \emph{periodic with period $d>1$} if $d$ is the common divisor of all the $t$ such that $A^{m+t}\sim A^{m}$ for a sufficiently large integer $m$. 
\end{definition}

Definition \ref{periodic} is a generalization of the definition of an irreducible periodic matrix {\cite[Def. 1.6]{Nonnegative_Matrices}}. In this definition, a periodic stochastic matrix is not necessarily irreducible. {With a slight abuse of terminology, we say the graph $	\mathcal{G}_W$ is \textit{periodic} if the associated matrix $W$ is periodic.}

In the context of distributed computation, it is always assumed that each individual computational unit in the network has access to its own latest state while implementing the iterative update rules \cite{optimization_1,Agreeing_asynchronously}. A class of situations that have received considerably less attention in the literature arise when some individuals are not able to obtain their own state, a case which can result from memory loss. Similar phenomena have also been observed in social networks while studying the evolution of opinions. Self-contemptuous people change their opinions solely in response to the opinions of others. The existence of computational units or individuals who are not able to access their own states sometimes might result in the computational failure or opinions' disagreement. 
As such an example, a periodic matrix $W$, which must has all zero diagonal entries (no access to their own states for all individuals), always leads the system \eqref{eq:syst_in} to oscillation. {This is because for a periodic $W$, $W^k$ never converges to a matrix with identical rows as $k\to\infty$. Instead, the positions of $W^{k}$ that have positive values are periodically changing with $k$, resulting in a periodically changing value of $W^kx(0)$}. This motivates us to investigate the particular case where $W$ is possibly periodic.

In this section, we show that agreement can be reached even when $W$ is periodic, \textit{just by introducing asynchronous updating events to the coupled agents}. In fact, perfect synchrony is hard to realize in practice as it is difficult for all agents to have access to a common clock according to which they coordinate their updating actions, while asynchrony is more likely. Researchers have studied how agreement can be preserved with the existence of asynchrony, see e.g., \cite{Chen_yao,xia_scrambling}. Unlike these works, we approach the same problem from a different aspect, where agreement occurs just because of asynchrony. A counterpart of this problem where $W$ is irreducible and periodic has been covered in our earlier work \cite{Qin2017}. We consider a more general case in this section where $W$ can be reducible.

To proceed, we define a framework of randomly asynchronous updating events. It is usually legitimate to postulate that on occasions more than one, but not all, agents may update.  Assume that each agent is equipped with a clock, which need not be synchronized with other clocks. The state of each agent remains unchanged except when an activation event is triggered by its own clock. Denote the set of event times of the $i$th agent by $\mathcal{T}^i=\{0,t^i_1,\cdots,t^i_k,\cdots\},k \in\N$. At the event times, agent $i$ updates its state obeying the asynchronous updating rule 
\begin{equation}\label{eq:indi_update}
{x_i}\left( {t_{k + 1}^i} \right) = \sum\limits_{j = 1}^n {{w_{ij}}{x_j}\left( {t_k^i} \right)} ,
\end{equation}
where $i\in\mathbf{N}$. We assume that the clocks which determine the updating events for the agents are driven by an underlying random process. The following assumption is important for the analysis.

\begin{assumption}\label{asynchronous_updating}
	 For any agent $i$, the intervals between two event times, denoted by $h^i_k=t^i_k-t^i_{k-1}$, are such that
	\begin{enumerate}
		\item[(i)] $h^i_k$ are upper bounded with probability 1 for all $k$ and all $i$;
		\item[(ii)] $\{h^i_k:k\in\N_0\}$ is a random sequence, with $\{h^1_k\}$, $\{h^2_k\}$, $\dots$, $\{h^n_k\}$ being \textit{mutually independent}. 
	\end{enumerate}
\end{assumption}

Assumption \ref{asynchronous_updating} ensures that an agent can be activated again within finite time after it is activated at $t^i_{k-1}$ for all $k \in \N$, which implies that all agents will update their states for infinitely many times in the long run.	
In fact, Assumption \ref{asynchronous_updating} can be satisfied if the agents are activated by mutually independent Poisson clocks or at rates determined by mutually independent Bernoulli processes {(\cite[Ch. 6]{Cassandras}, {\cite[Ch. 2]{Gallager}})}.

Let $\mathcal{T}=\{t_0,t_1,t_2,\cdots,t_k,\cdots\}$ denote all event times of all the $n$ agents, in which the event times have been relabeled in a way such that $t_0=0$ and $t_{\tau}<t_{\tau+1},\tau=\{0,1,2,\cdots\}$. This idea has been used in \cite{Bertsekas} and \cite{Agreeing_asynchronously} to study asynchronous iterative algorithms. One situation may occur in which there exist some $k$ such that $t_k \in \mathcal{T}^i$ and $t_k \in \mathcal{T}^j$ for some $i,j$, which implies more than one agent is activated at some event times. Although this is not likely to happen when the underlying process is some special random ones like Poisson, our analysis and results will not be affected. For simplicity, we rewrite the set of event times as $\mathcal{T}=\{0,1,2,\cdots,k,\cdots\}$. Then the system with asynchronous updating can be treated as one with discrete-time dynamics in which the agents are permitted to update only at certain event times $k,k\in \N$, according to the updating rule \eqref{eq:indi_update}  at each time $k$. Since each $k \in \mathcal{T}$ can be the event time of any subset of agents, we can associate any set of event times $\{k+1,k+2,\dots,k+h\}$ with the updating sequence of agents $\{\lambda(k+1),\lambda(k+2),\dots,\lambda(k+h)\}$ with $\lambda(i)\in\mathcal{V}$. Under Assumption \ref{asynchronous_updating}, one knows that this updating sequence can be \textit{arbitrarily ordered}, and each possible sequence can occur with positive probability, though the particular value is not of concern.

Assume at time $k$, $m\ge 1$ agents are activated, labeled by $k_1, k_2, \dots,$ $k_m$,  then we define the following matrices 
\begin{equation} \label{eq:Up_Ma}
{W(k)} = {\left[ {{u_1}, \cdots ,w_{{k_1}}^\top,{u_{k + 1}}, \cdots ,w_{{k_m}}^\top, \cdots ,{u_n}} \right]^\top}, 
\end{equation}
where $u_i \in \mathbb{R}^n $ is the $i$th column of the identity matrix $I_n$ and $w_k \in \mathbb{R}^{ n} $ denotes the $k$th row of $W$. We call $W(k)$ the \textit{asynchronous updating matrix} at time $k$. Then the asynchronous updating rule \eqref{eq:indi_update} becomes
\begin{equation} \label{eq:asy}
x_k=W(k)x_{k-1}, k\in\N,
\end{equation}
where $\{W(k)\}$ is a random sequence of asynchronous updating matrices which are stochastic, and $x_0\in\mathbb{R}^n$ is a given initial state. We say the \textit{asynchronous agreement} is   reached if $x_k$ converges to a scaled all-one vector when the agents update asynchronously.
It suffices to study the convergence of the product $W(k)\dots W(2)W(1)$ to a rank-one matrix. We now show the asynchronous agreement is reached almost surely even when the graph is periodic. A necessary and sufficient condition for the graph is obtained, under which the agreement can always be reached. 
%
%

\begin{theorem}\label{thm_indecom}
	If the agents coupled by a network update asynchronously under Assumption \ref{asynchronous_updating}, they reach agreement \textit{almost surely} if and only if the network is rooted, i.e., the matrix $W$ is indecomposable.
\end{theorem}

To prove this theorem, we need to introduce some additional concepts and results. It is equivalent to say the associated graph $\mathcal{G}_W$ is rooted if $W$ is indecomposable. Denote the set of all the roots of $\mathcal{G}_W$ by $\mathbf{r}\subseteq\mathcal{V}$. We can partition the vertices of $\mathcal{G}_W$ into some hierarchical subsets as follows. 
For any $\kappa\in \mathbf{r}$, there must exist at least one directed spanning tree rooted at $\kappa$, see e.g., Fig. \ref{spanning_tree} (a). We select any of these directed spanning trees, denoted by $\mathcal{G}^s_W$. There exists a directed path from $\kappa$ to any other vertex $i \in \mathcal{V}\backslash \kappa$, see e.g., Fig. \ref{spanning_tree} (b). Let $l_i$ be the length of the directed path from $\kappa$ to $i$, and  there exists an integer $L\le n$ such that $l_i < L$ for all $i$. Define
\begin{align*}{\mathcal{H}_r}: = \left\{ {i:{l_i} = r} \right\},r = 1, \cdots ,L-1,\end{align*}
and $\mathcal{H}_0=\{\kappa\}$. From this definition, one can partition the vertices of $\mathcal{G}^s_W$ into $L$ hierarchical subsets, i.e., $\mathcal{H}_0,\mathcal{H}_1,\cdots,\mathcal{H}_{L-1}$, according to the vertices' distances to the root $\kappa$.  Let $n_r$ be the number of vertices in the subset $\mathcal{H}_r$, $0 \le r \le L-1$ (see the example in Fig. \ref{spanning_tree} (b)). Note that given a spanning tree, its corresponding hierarchical subsets $H_r$'s are uniquely determined.

\begin{figure}[!t]
	\centering
	
	\subfigure[The original graph.]{	
		\begin{tikzpicture} [->,>=stealth',shorten >=1pt,auto,node distance=1.2cm,
		main node/.style={circle,fill=blue!10,draw,minimum size=0.4cm,inner sep=0pt]}]
		
		\node[main node]          (3)     at (0,0)                     {$3$};
		\node[main node]          (2) at (-1,-1)          {$2$};
		\node[main node]          (6) at (1,-1)         {$6$};
		\node[main node]          (4) at (0.2,-1.8)          {$4$};
		\node[main node]          (5) at (1,-2.6)           {$5$};
		\node[main node]          (1) at (0,-3.5)          {$1$};
		
		\tikzset{direct/.style={->,line width=1pt}}
		\path (3)     edge [->,in=75,out=190,dashed,line width=1pt,red]    node   {} (2) 
		(2)     edge [->,in=250,out=355,line width=1pt,black]    node   {} (3) 
		(3)     edge [direct,dashed,line width=1pt,red]    node   {} (6) 
		(6)     edge [->,in=65,out=190,dashed,line width=1pt,red]    node   {} (4) 
		(4)     edge [->,in=155,out=280,,dashed,line width=1pt,red]    node   {} (5) 
		(5)     edge [->,in=330,out=45,line width=1pt]    node   {} (6) 
		(5)     edge [direct,line width=1pt]    node   {} (1) 
		(2)     edge [direct,dashed,line width=1pt,red]    node   {} (1) ;
		\end{tikzpicture}}$\;\;\;\;\;$
	\subfigure[Partition of the vertices.]{	\label{Fig:spaning}
		\begin{tikzpicture} [->,>=stealth',shorten >=1pt,auto,node distance=1.2cm,
		main node/.style={circle,fill=blue!10,draw,minimum size=0.4cm,inner sep=0pt]}]
		
		\filldraw[dashed,fill=orange!5] (-1.5,-1.3)   rectangle (1.5cm,-0.7cm); 
		\filldraw[dashed,fill=orange!5] (-1.5,-2.5)   rectangle (1.5cm,-1.8cm);
		\node[main node]          (3)  at (0,0)                      {$3$};
		\node[main node]          (2) at (-1,-1)           {$2$};
		\node[main node]          (6) at (1,-1)           {$6$};
		\node[main node]          (4) at (1,-2.2)           {$4$};
		\node[main node]          (5) at (1,-3.5)          {$5$};
		\node[main node]          (1) at (-1,-2.2)          {$1$};
		
		\node[] ()  at (2.0,0) {$\mathcal{H}_0$};
		\node[] ()  at (2.0,-1) {$\mathcal{H}_1$};
		\node[] (7)  at (2.0,-2.2) {$\mathcal{H}_2$};
		\node[] (7)  at (2.0,-3.5) {$\mathcal{H}_3$};
		
		\tikzset{direct/.style={->,line width=1pt}}
		\path (3)     edge [direct,dashed,line width=1pt,red]    node   {} (2) 
		(3)     edge [direct,dashed,line width=1pt,red]    node   {} (6) 
		(6)     edge [direct,dashed,line width=1pt,red]    node   {} (4) 
		(4)     edge [direct,dashed,line width=1pt,red]    node   {} (5)
		(2)     edge [direct,dashed,line width=1pt,red]    node   {} (1) ;
		
		\end{tikzpicture}}
	
	\caption{An illustration of the graph partition; the hierarchical subsets: $\mathcal{H}_0=\{3\}$,$\mathcal{H}_1=\{2,6\}$,$\mathcal{H}_2=\{1,4\}$,$\mathcal{H}_3=\{5\}$; for example, \{3,2,6,1,4,5\} is a hierarchical updating vertex sequence. } 
	\label{spanning_tree}
\end{figure}
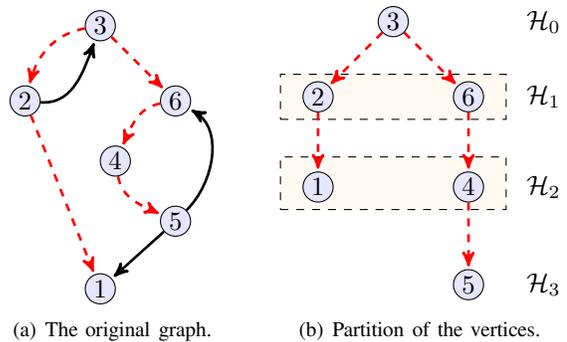

\begin{definition}\label{hierarchical_sequence}
	An {updating vertex sequence} of length $n$ is said to be hierarchical if it can be partitioned into some successive subsequences, denoted by $\{\mathcal{A}_0,\dots,\mathcal{A}_{L-1}\}$ with $\mathcal{A}_r=\{\lambda_r(1),\lambda_r(2),\cdots,\lambda_r(n_r)\}$, such that $\bigcup\nolimits_{k = 1}^{{n_r}} {{\lambda _r}\left( k \right) = {\mathcal{H}_r}}$ for all $r=0,\cdots,L-1$, where $\mathcal{H}_r$'s are the hierarchical subsets of some spanning tree $\mathcal{G}^s_W$ in $\mathcal{G}_W$ .
\end{definition}

\begin{proposition}\label{Markov_hierarchical}
	If agents coupled by $\mathcal{G}_W$ update in a hierarchical sequence $\{a_1,\cdots,a_n\},a_i \in \mathcal{V}$ for all $i$, the product of the corresponding asynchronous updating matrices, $\tilde{W}:={W_{{a_n}}} \cdots {W_{{a_2}}}{W_{{a_1}}}$,	is a Markov matrix.
\end{proposition}

To prove this proposition, we define an operator $\mathcal{N}(\cdot ,\cdot)$ for any stochastic matrix and any subset $\mathcal{S}\in\mathcal{V}$
\[
\mathcal{N}(A,\mathcal{S}):=\{j:A_{ij}>0,i\in \mathcal{S}\},
\]
and we write $\mathcal{N}(A,\{i\})$ as $\mathcal{N}(A,i)$ for brevity. It is easy to check then for any two stochastic matrices $A_1,A_2 \in \mathbb{R}^{n\times n}$ and for any subset  $\mathcal{S} \in \mathcal{V}$, it holds that
\begin{equation}\label{property_neighbor}
	{\cal N}\left( {{A_2}{A_1},{\cal S}} \right) = {\cal N}\left( {{A_1},{\cal N}\left( {{A_2},{\cal S}} \right)} \right).
\end{equation}

\begin{IEEEproof}[Proof of Proposition \ref{Markov_hierarchical}] 
	It suffices to show that all $i\in\mathcal{V}$  share at least one common neighbor in the graph $\mathcal{G}_{\tilde{W}}$, i.e.,
	\begin{equation}\label{intersection_set}
		\bigcap\nolimits_{i = 1}^n {{\cal N}\left( {\tilde W,i} \right)}  \ne \emptyset .\end{equation}
	We rewrite the product of asynchronous updating matrices into
	\[\tilde W = \left\{ {{W_{{\lambda _{L - 1}}\left( 1 \right)}} \cdots {W_{{\lambda _{L - 1}}\left( {{n_{L - 1}}} \right)}} \cdots {W_{{\lambda _{L - 2}}\left( 1 \right)}} \cdots {W_{{\lambda _0}\left( 1 \right)}}} \right\}.\]
	For any distinct $i,j\in\mathcal{V}$, we know that $\mathcal{N}(W_j,i)=\{i\}$ from the definition of asynchronous updating matrices. 
	Then for any $\lambda_r(t)\in \mathcal{H}_r, t\in\{1,\cdots,n_r\}, r\in\{1,\cdots,L-1\}$, it holds that
	\[\begin{array}{*{20}{l}}
	{{\cal N}\left( {\tilde W,{\lambda _r}(t)} \right)}\\
	{ = {\cal N}\left( {{W_{{\lambda _r}\left( t \right)}}{W_{{\lambda _r}\left( {t + 1} \right)}} \cdots {W_{{\lambda _r}\left( {{n_r}} \right)}} \cdots {W_{{\lambda _0}\left( 1 \right)}},{\lambda _r}\left( t \right)} \right) }\\
	={{\cal N}\left( {{W_{{\lambda _r}\left( {t + 1} \right)}} \cdots {W_{{\lambda _r}\left( {{n_r}} \right)}} \cdots {W_{{\lambda _0}\left( 1 \right)}},{\cal N}\left( {{W_{{\lambda _r}\left( t \right)}},{\lambda _r}\left( t \right)} \right)} \right),}
	\end{array}\]
	where the property \eqref{property_neighbor} has been used.	From Definition \ref{hierarchical_sequence}, one knows that there exists at least one vertex ${\lambda _{r-1}}\left( t_1 \right)\in \mathcal{H}_{r-1}$ that can reach ${\lambda _r}\left( t \right)$  in  $\mathcal{G}_W$ and subsequently in $\mathcal{G}_{W_{{\lambda _r}\left( t \right)}}$, which implies 
	\[{\lambda _{r - 1}}\left( {{t_1}} \right) \in {\cal N}\left( {{W_{{\lambda _r}\left( t \right)}},{\lambda _r}\left( t \right)} \right).\]
	It then follows 
	\[\begin{array}{l}
	\mathcal{N}\left( {{W_{{\lambda _r}\left( {t + 1} \right)}} \cdots {W_{{\lambda _r}\left( {{n_r}} \right)}} \cdots {W_{{\lambda _0}\left( 1 \right)}},{\lambda _{r - 1}}\left( {{t_1}} \right)} \right)\\
	\subseteq \mathcal{N}\left( {\tilde W,{\lambda _r}(t)} \right).
	\end{array}\]
	Similarly, there hold that 
	\[\begin{array}{l}
	{\cal N}\left( {{W_{{\lambda _r}\left( {t + 1} \right)}} \cdots {W_{{\lambda _r}\left( {{n_r}} \right)}} \cdots {W_{{\lambda _0}\left( 1 \right)}},{\lambda _{r - 1}}\left( {{t_1}} \right)} \right)\\
	= {\cal N}\left( {{W_{{\lambda _{r - 1}}\left( {{t_1}} \right)}} \cdots {W_{{\lambda _r}\left( {{n_r}} \right)}} \cdots {W_{{\lambda _0}\left( 1 \right)}},{\lambda _{r - 1}}\left( {{t_1}} \right)} \right)\\
	= {\cal N}\left( {{W_{{\lambda _{r - 1}}\left( {{t_1} + 1} \right)}} \cdots {W_{{\lambda _0}\left( 1 \right)}},{\cal N}\left( {{W_{{\lambda _{r - 1}}\left( {{t_1}} \right)}},{\lambda _{r - 1}}\left( {{t_1}} \right)} \right)} \right)\\
	\supseteq {\cal N}\left( {{W_{{\lambda _{r - 1}}\left( {{t_1} + 1} \right)}} \cdots {W_{{\lambda _0}\left( 1 \right)}},{\lambda _{r - 2}}\left( {{t_2}} \right)} \right).
	\end{array}\]
	As a recursion, it must be true that 
	\begin{equation}\label{subset_1}
		{\cal N}\left( {{W_{{\lambda _0}\left( 1 \right)}},\kappa} \right) \subseteq {\cal N}\left( {\tilde W,{\lambda _r}(t)} \right),	
	\end{equation}
	where $\kappa$ is a root of $\mathcal{G}_W^s$. In fact, it holds that $\lambda_0(1)=\kappa$, and then we know
	\begin{equation}\label{subset_f}
		{\cal N}\left( {{W_{{\lambda _0}\left( 1 \right)}},\kappa} \right) = {\cal N}\left( {{W_\kappa},\kappa} \right) = {\cal N}\left( {W,\kappa} \right).
	\end{equation}
	Substituting \eqref{subset_f} into \eqref{subset_1} leads to
	\[{\mathcal N}\left( {W,\kappa} \right) \subseteq {\mathcal N}\left( {\tilde W,{\lambda _r}(t)} \right)\]
	for all ${\lambda _r}(t)$. Since $\bigcup\nolimits_{r,t} {\left\{ {{\lambda _r}\left( t \right)} \right\}}  = \mathcal{V}$, we know
	\[{\cal N}\left( {W,\kappa} \right) \subseteq \bigcap\nolimits_{r,t} {{\cal N}\left( {\tilde W,{\lambda _r}(t)} \right)}. \]
	Straightforwardly, \eqref{intersection_set} follows, which completes the proof.
\end{IEEEproof}

Since the hierarchical sequences will appear with positive probability in any sequence of length $n$, one can easily prove the following proposition by letting $l=n$.
\begin{proposition}\label{periodic_pro}
	There exist an integer $l$ such that the product $W(k+l)\cdots W(k+1)$, where $W(k)$ is given in \eqref{eq:asy}, is a Markov matrix with positive probability for any $k\ge0$. 
\end{proposition}

\begin{IEEEproof}[Proof of  Theorem \ref{thm_indecom}] We prove the necessity by contradiction. Suppose the matrix $W$ is decomposable. Then there are at least two sets of vertices that are isolated from each other. Then agreement will never happen between these two isolated groups if they have different initial states.
Let $l=n$, in view of Corollary \ref{scrambling_cora}, the sufficiency follows directly from Proposition \ref{periodic_pro}, which completes the proof.
\end{IEEEproof} 

\begin{remark}	
Note that the hierarchical sequence is a particular type of updating orders that results in a Markov matrix as the product of the corresponding updating matrices. We have identified another type of updating orders in our earlier work when $W$ is irreducible and periodic \cite{Qin2017}.  It is of great interest for future work to look for other updating mechanisms to enable the appearance of Markov matrices or scrambling matrix to guarantee asynchronous agreement.
\end{remark}
In the next section, we look into another application in solving linear algebraic equations.

\section{To Solve Linear Algebraic Equations}

Researchers have been quite interested in solving a system of linear algebraic equations in the form of $Ax=b$ in a distributed way \cite{Lu_j,Lu_j_1,Liu_Ji_2017,Mou_2016}. In this section we deal with the problem under the assumption that this system of equations has at least one solution. The set of equations is decomposed into smaller sets and distributed to a network of $n$ processors, referred to as agents, to be solved in parallel. Agents can receive information from their neighbors and the neighbor relationships are described by a time-varying $n$-vertex directed graph $\mathcal{G}(t)$ with self-arcs.  When each agent knows only the pair of real-valued matrices $(A_i^{n_i \times m},b_i^{n_i\times 1})$, the problem of interest is to devise local algorithms such that all $n$ agents can  iteratively compute the same solution to the linear equation $Ax=b$, where $A=[A_1^\top,A_2^\top,\dots,A_n^\top]^\top, b=[b^\top_1,b^\top_2,\dots,b^\top_n]^\top$ and $\sum_{i=1}^{n}n_i=m$. A distributed algorithm to solve the problem is introduced in \cite{equation_Mou}, where the iterative updating rule for each agent $i$ is described by
\begin{align}
x^ i_{k+1} = {x^i_k} - {\frac{1}{d^ i_k}} {P_i}\Big( d^ i_k x^ i_k - \sum\limits_{j \in {{\cal N}_i}(k)} x^j_k \Big), 
k \in \N,\label{solve_equation}
\end{align}
where $x^ i_k \in \mathbb{R}^{m}$, $d^ i_k$ is the number of neighbors of agent $i$ at time $k$, ${\cal N}_i(k)$ is the collection of $i$'s neighbors, $P_i$ is the orthogonal projection on the kernel of $A_i$, and the initial value $x^ i_1$ is any solution to the equations of $A_ix=b_i$.

The results in \cite{equation_Mou} have shown that all $x^i_k$ converge to the same solution  exponentially fast if the sequence of graphs $\mathcal{G}(t)$ is repeatedly jointly strongly connected. This condition is restrictive  since it is required that for some integer $l$, the composition of the sequence of graphs, $\{\mathcal{G}(k),\dots,\mathcal{G}(k+l-1)\}$, must be strongly connected for any $t$. By the composition of a directed graph $\mathcal{G}_1$ with the vertex set {$\mathcal V$} with another directed graph $\mathcal{G}_2$ with the same vertex set {$\mathcal V$}, denoted by $\mathcal{G}_2 \circ \mathcal{G}_1$, we mean the directed graph with the vertex set {$\mathcal V$} and edge set defined in such a way that $(i,j)$ is an arc of the composition just in case there is a vertex $i_1$ such that $(i,i_1)$ is an edge in $\mathcal{G}_1$ and meanwhile $(i_1,j)$ is an edge in $\mathcal{G}_2$.
It is not so easy to satisfy this condition if the network is changing randomly. Now assume that the evolution of the sequence of graphs $\{\mathcal{G}(1),\dots,\mathcal{G}(k),\dots\}$ is driven by a random process. In this case, results in Theorem \ref{stability_finite} and Corollary \ref{Coroll_invariant} can be applied to relax the condition in \cite{equation_Mou} to achieve the following more general result.

\begin{theorem}\label{theorem_solving-equation}
	Suppose each agent updates its state $x^i_k$ according to the rule \eqref{solve_equation}. All states $x^i_k$ converge to the same solution to $Ax=b$ almost surely if the following two conditions are satisfied
	\begin{itemize}
		\item[a)] there exists an integer $l$ such that the composition of any sequence of randomly changing graphs $\{\mathcal{G}(k),\mathcal{G}(k+1),\dots,\mathcal{G}(k+l-1)\}$ is strongly connected with positive probability $p(k)>0$ for any $k \in \N$;
		\item[b)] there holds $\sum\nolimits_{i = 0}^\infty  {p\left( {{k} + il} \right)}  = \infty, \forall k.$
	\end{itemize}
\end{theorem}

To prove the theorem, we define an error system. Let $x^*$ be any solution to $Ax=b$, so $A_ix^*=b_i$ for any $i$. Then, we define 
\[
e^i_k=x^i_k-x^*,i \in \mathcal V, k\in \N,
\]
which, as is done  in \cite{equation_Mou}, can be simplified into
\begin{equation}\label{error_dis}
e^i_{k+1}=\frac{1}{d^ i_k}P_i\sum_{j\in\mathcal{N}_i(k)}P_j e^j_{k}.
\end{equation}
Let $e_k=[{e^1_k}^ \top,\dots,{e^n_k}^ \top]^\top$, $A(k)$ be the adjacency matrix of the graph $\mathcal{G}(k)$, $D(k)$ be the diagonal matrix whose $i$th diagonal entry is $d^ i_k$, and $W(k)=D^{-1}(k)A^\top(k)$. It is clear that $W(k)$ is a stochastic matrix, and $\{W(k)\}$ is a stochastic process. Now we  write equation \eqref{error_dis} into a compact form
\begin{equation}\label{compact_error}
e_{k+1}=P(W(k)\otimes I)Pe_k,k\in \N,
\end{equation}
where $\otimes$ denotes the Kronecker product, $P:={\rm diag} \{P_1,P_2,$ $\dots$, $P_n\}$, and $\{W(k)\}$ is a random process. We will show this error system is {globally a.s. asymptotically stable}. Define the transition matrix of this error system by 
\[
\Phi(k+T,k)=P(W(k+T-1)\otimes I)P\cdots P(W(k)\otimes I)P.
\]
In order to study the stability of the error system \eqref{compact_error}, we define a mixed-matrix norm for an $n \times n$ block matrix $Q=[Q_{ij}]$ whose $ij$th entry is a matrix $Q_{ij} \in \mathbb{R}^ {m \times m}$, and 
\[
\BlockNorm{Q} = {\left| {\left\langle Q \right\rangle } \right|_\infty },\]
where ${\left\langle Q \right\rangle }$ is the matrix in $\mathbb{R}^{n\times n}$ whose $ij$th entry is $|Q_{ij}|_2$. Here  $\| \cdot \|_2$ and $\| \cdot \|_\infty$ denote the induced 2 norm and infinity norm, respectively. It is easy to show that $\BlockNorm{\cdot}$ is a norm. Since $\|Ax\|_2\le \|A\|_2\|x\|_2$ for $x\in \mathbb{R}^{nm \times nm}$, it follows straightforwardly that $\BlockNorm{Ax} \le \BlockNorm{A}\BlockNorm{x}$. It has been proven in \cite{equation_Mou} that $\Phi(k+T,k)$ is non-expansive for any $k>0,T\ge0$. In other words, it holds that $\BlockNorm{\Phi(k+T,k)}\le1$. Moreover, the transition matrix is a contraction, i.e., $\BlockNorm {\Phi(k+T,k)}< 1$, if there exists a ``route" $j=i_0,i_1,\dots,i_T=i$ over the sequence $\{\mathcal{G}(k),\dots,\mathcal{G}(k+T-1)\}$ for any $i,j\in \mathcal V$ that satisfies $\bigcup\nolimits_{k = 0}^T {\left\{ {{i_k}} \right\}}  = \mathcal V$; here by a \emph{route} over a given sequence of graphs $\{\mathcal{G}(1),\mathcal{G}(2),\dots,\mathcal{G}(k)\}$, we mean a sequence of vertices $i_0,i_1,\dots,i_k$ such that $(i_{j-1},i_j)$ is an edge in $\mathcal{G}(z)$ for all $1\le z \le k$. Now we are ready to prove Theorem \ref{theorem_solving-equation}.
\begin{IEEEproof}[Proof of Theorem \ref{theorem_solving-equation}]
	Let $V(e_k)=\BlockNorm{e_k}$ be a finite-step stochastic Lyapunov function candidate. Let $\{\mathcal F_k\}$, where $\mathcal{F}_k= \sigma(\mathcal{G}(1), \cdots, \mathcal{G}(k), \cdots)$, be an increasing sequence of $\sigma$-fields. We first show that $V(e_k)$ is a supermartingale with respect to $\mathcal F_k$ by observing
	\begin{align*}
	\mathbb{E}\big[ {\left. {V\big( {e_{k + 1} } \big)} \right|{\mathcal{F}_k}} \big] = \mathbb{E}  \BlockNorm{{{\Phi _k}e_k}} 
	\le \mathbb{E}\BlockNorm{{\Phi _k}}\BlockNorm{e_k}\le \BlockNorm{e_k},
	\end{align*}
	where $\Phi_k=\Phi(k,k)=P(W(k)\otimes I)Pe_k$. The last inequality follows from the fact that $\mathbb{E}\BlockNorm{\Phi_k}\le 1$ since all the possible $\Phi_k$ are non-expansive. Consider the sequence of randomly changing graphs $\{\mathcal{G}(1),\mathcal{G}(2),\cdots,\mathcal{G}(q)\}$, where $q=(n-1)^2l$. Let $r=n-1$, and partition this sequence into $r$ successive subsequences $\mathcal{G}_1=\{\mathcal{G}(1),\dots,\mathcal{G}(rl)\}$, $\mathcal{G}_2=\{\mathcal{G}(rl+1),\dots,\mathcal{G}(2rl)\}$,$\cdots$, $\mathcal{G}_r=\{\mathcal{G}((r-1)l+1),\dots,\mathcal{G}(r^2l)\}$. Let $\mathbb{C}_z$ denote the composition of the graphs in the $z$th subsequence, i.e., $\mathbb{C}_z = \mathcal{G}\left( {zl} \right) \circ  \cdots  \circ \mathcal{G}\left( {(z - 1)l + 2} \right) \circ \mathcal{G}\left( {(z - 1)l + 1} \right),z=1,2,\dots,r$. Since all the subsequences have the length $rl$, each can be further partitioned into $r$ successive sub-subsequences of length $l$. From the condition of Theorem \ref{theorem_solving-equation}, one knows that the composition of the graphs in any sub-subsequence has positive probability to be strongly connected. The event that the composition of the graphs in each of the $r$ sub-subsequences in $\mathcal{G}_z$  is strongly connected also has  positive probability. This holds for all $z$. We know that the composition of any $r$ or more strongly connected graphs, within which each vertex has a self-arc, results in a complete graph \cite{Reaching_consensus}. It follows straightforwardly that the graphs $\mathbb{C}_1,\cdots,\mathbb{C}_r$ have positive probability to be all complete. Therefore, for any pair $i,j\in \mathcal V$, there exists a route from $j$ to $i$ over the graph $\mathbb{C}_z$ for any $z$. It is easy to check that there exists a route $i_1,i_2,\dots,i_n$ over the graphs $\mathbb{C}_1,\cdots,\mathbb{C}_r$, where $i_1,i_2,\dots,i_n$ can be any reordered sequence of $\{1,2,\dots,n\}$. Similarly, for any $x$ there must exist a route of length $rl$, $i_z=i_z^1,i_z^2,\dots,i_z^{rl}=i_{z+1}$, over $\mathcal{G}_z$. Thus there is a route $i_1^1,i_1^2, \ldots ,i_1^{rl},i_2^2, \ldots ,i_2^{rl} \ldots ,i_r^{rl}$ over the graph sequence $\{\mathcal{G}(1),\mathcal{G}(2),\cdots,\mathcal{G}(q)\}$ so that $\bigcup\nolimits_{\delta=1} ^r {\bigcup\nolimits_{\theta  = 1}^{rl} {\left\{ {i_\delta ^\theta } \right\}} }  = {\mathcal V}$. This implies that the probability that $\Phi(q,1)$ being a contraction is positive. Since all $\Phi(q,1)$ are non-expansive,  there is a number $\rho(1)<1$ such that $\mathbb{E} \BlockNorm{\Phi(q,1)}=\rho(1)$. Straightforwardly, it also holds $\mathbb{E}\BlockNorm{\Phi(k+q,k)}=\rho(k)<1$ for all $k<\infty$. Thus there a.s. holds that
	\begin{align*}
	\mathbb{E}\big[ {\left. {V\left( {e_{k+q}} \right)} \right|{\mathcal{F}_k}} \big] &- V(e_k)= 	\mathbb{E} \BlockNorm{{\Phi \left( {k + q,k} \right)}e_k}  - V(e_k)\\
	&\le \mathbb{E} \BlockNorm{\Phi \left( {k + q,k} \right)} \cdot \BlockNorm {e_k} - V(e_k)\\
	&= ( \rho (k)-1)V(e_k).
	\end{align*}	
	Similarly as in the proof of Theorem  \ref{theorem_geral}, the condition b) in Theorem \ref{theorem_solving-equation} ensures that $\sum\nolimits_{i=1}^{\infty}(1-\rho(k))=\infty$. It follows that $V\left(e_k\right)\stackrel{a.s.}{\longrightarrow}0$ as $t\to \infty$ since $V(e_0)-\mathbb{E}\big[ {\left. {V\left( {e_{nq}} \right)} \right|{\mathcal{F}_k}} \big] < \infty $ for any $N$. Define the set $\mathcal{Q}:=\{e:V(e)\le V(e_1)\}$ for any initial $e_1$ corresponding to $x_1$. For any random sequence $\{\mathcal{G}(k)\}$, it follows from the system dynamics \eqref{compact_error} that 
	\begin{equation}
	V(e_k)\le V(e_{k-1})\cdots \le V(e_{2}) \le V(e_{1}),\nonumber
	\end{equation}
	and thus $e_k$ will stay within the set $\mathcal{Q}$ with probability $1$. From Theorem \ref{stability_finite} and Corollary \ref{Coroll_invariant}, it follows that $e_k$  asymptotically converges to  $\{e:V(e)=0\}$ almost surely.  Moreover, since $V(e)$ is a norm of $e$, it can be concluded  {from Corollary \ref{Coroll_invariant}} that the error system \eqref{compact_error} is {globally a.s. asymptotically stable}. The proof is complete.	
\end{IEEEproof}
 It is worth mentioning  that the error system is {globally a.s. exponentially stable} under the assumption that the probability of the composition of any sequence of randomly-changing graphs, $\{\mathcal{G}(k),\dots,\mathcal{G}(k+1),\mathcal{G}(k+l-1)\}$, for any $k\ge 0$, being strongly connected is lower bounded by some positive number. This can be proven with the help of Theorem \ref{stability_expone_finite} and Corollary \ref{Coroll_invariant_expo}. 

\section{Concluding Remarks}
We have established the tool of finite-step stochastic Lyapunov functions, using which one can study the {convergence and stability} of a stochastic system together with its convergence rate. As  applications, we investigate the convergence of the products of a random sequence of stochastic matrices. The asynchronous agreement problem and the distributed algorithm for solving linear algebraic equations have also been studied. Conditions in the existing results on both of these problems have been relaxed. One of our future research directions is to apply finite-step stochastic Lyapunov functions to the study of stochastic distributed optimization. 


\section{Acknowledgement}
We thank Prof. Tobias M\"uller from  Bernoulli Institute, University of Groningen, for constructive discussions.

\appendices
\section{An Alternative Proof of Corollary \ref{stationary_ergodic}}\label{appendix_1}
%

For ergodic stationary sequences, the following important property is the key to construct the convergence rate.

\begin{lemma}[{Birkhoff's Ergodic Theorem, see \cite[Th. 7.2.1]{prob_book}}] \label{Birkhoff's} 
	For an ergodic sequence $\{X_k\},k \in \N_{\ge0}$, of random variables, it holds that 
	\begin{equation}
	\lim\limits_{m \to \infty}\frac{1}{m}\sum_{k=0}^{m-1}X_k \stackrel{a.s.}{\longrightarrow} E(X_0)
	\end{equation}
\end{lemma}

For the product given in \eqref{product}, we say $W(k,0)$ converges to a rank-one matrix $W=1\xi^\top$ a.s. as $k\to\infty$ if  $\tau(W(k,0))\to 0$ as $k\to \infty$, where $\tau(\cdot)$ is defined in \eqref{coeffi_ergo}. According to Definition \ref{def:convergence}, if there exist $\beta>1$ such that 
\begin{equation}
\beta^k\tau\big(W(k,0)\big) \stackrel{a.s.}{\longrightarrow} 0, k\to\infty,
\end{equation} 
then the convergence rate is said to be exponential at the rate no slower than $\beta^{-1}$. We are now ready to present the proof of Corollary \ref{stationary_ergodic}.

\begin{proof}[Proof of Corollary \ref{stationary_ergodic}]
	Let $h$ be the same as that in Assumption \ref{assumption_2}. There is an integer $\theta\in\N$ such that $W(t+\theta h,t)$ is scrambling with positive probability. Let $T=\theta h$. Consider a sufficiently large $r$, and then $W(r,0)$ can be written as
	\[W(r,0)) = \bar W \cdot W\big( {mT,\left( {m - 1} \right)T} \big) \cdots W\left( {T,{0}} \right),\]
	where $m$ is the largest integer such that $mT\le r$, $W \left( {kT+T,kT}\right),k=0,\cdots,m-1$, are the matrix products defined by \eqref{product}, and ${\bar W}=W(r,mT)$ is the remaining part, which is obviously a stochastic matrix. To study the limiting behavior of $W(r,0)$, we compute its coefficients of ergodicity
	\begin{align*}
	\tau \big( {W\left( {r,0} \right)} \big) \le \tau\left( {\bar W} \right)&\prod\limits_{k = 0}^{m-1} \tau \big( W\left( {kT+T,kT } \right) \big)\\
	&\le \prod\limits_{k = 0}^{m-1} \tau \big( W\left( {kT+T,kT } \right) \big),	
\end{align*}
	where the property \eqref{coeffi_submulti} has been used. The last inequality follows from the property of coefficients of ergodicity, i.e., $\tau(A)\le 1$ for a stochastic matrix $A$. 
	Taking logarithms yields that
	\begin{equation}\label{log_coeffi}
	\log \tau \big( {W\left( {r,0} \right)} \big) \le \sum\limits_{k = 0}^{m-1} {\log \tau\big( W\left( {KT+T,kT} \right) \big)}.
	\end{equation}
	Since  the sequence $\{W(k)\}$ is ergodic, it is easy to see that the sequence of products $\{W\left( {kT+T,kT} \right)\}$, $k=0,\cdots,m-1$, over non-overlapping intervals of length $T$, is also ergodic. It follows in turn that $\{\log \tau\big( W\left( {kT+T,kT} \right) \big)\}$ is ergodic. From Lemma \ref{Birkhoff's}, one can further obtain
	\begin{align*}
	\lim\limits_{m\to\infty}\frac{1}{m}\sum\limits_{k = 0}^{m - 1} &\log \tau \big( W\left( {kT+T,kT } \right) \big)\\
	&\stackrel{a.s.}{\longrightarrow} \mathbb{E}\Big[ {\log \tau\big( {W \left( {{T},0} \right)} \big)} \Big]
	\le \log \mathbb{E}\big[\tau\big( {W \left( {{T},0} \right)} \big)\big].
	\end{align*}
	The last inequality follows from Jensen's inequality (see {\cite[Th. 1.5.1]{prob_book}}) since $\log(\cdot)$ is concave. According to Assumption \ref{assumption_1}, one knows that $W(t+h,t)$ is scrambling with positive probability, and thus it follows that $0<\mathbb{E}\big[\tau\left( {W \left( {{T},0} \right)} \right)\big]<1$. Taking a positive number $\lambda$ satisfying $\lambda<-\log \mathbb{E}\big[\tau\big( {W \left( {{T},0} \right)} \big)\big]$, one obtains 
	\[m\lambda + \sum\limits_{k = 0}^{m - 1} {\log \tau\big( {W \left( {{KT+T},k{T}} \right)} \big)}  \stackrel{a.s.}{\longrightarrow}  - \infty. \]
	Adding $m\lambda$ to both sides of \eqref{log_coeffi} yields that
	\begin{align*}
	m \lambda + \log \tau &\big( {W\left( {r,0} \right)} \big)\\
	&\le m \lambda + \sum\limits_{k = 0}^{m - 1} {\log \tau\big( {W \left( kT+T,kT \right)} \big)} \stackrel{a.s.}{\longrightarrow}  - \infty .
	\end{align*}
	 It follows straightforwardly that
	\[{\left( {{e^\lambda}} \right)^m}\tau \big( {W\left( {r,0} \right)} \big) \stackrel{a.s.}{\longrightarrow} 0.\]
	Let $\beta=e^\lambda$, which apparently satisfies $\beta>1$. From Definition \ref{def:convergence}, one can conclude that the product $W(k,0)$ almost surely converges to a rank-one stochastic matrix exponentially at a rate no slower than $\beta^{-1}$, which completes the proof.
\end{proof}

\bibliographystyle{IEEEtran}
\bibliography{IEEEabrv,stochastic_Lyapunov_arXiv}

\end{document}